\documentclass[lettersize,journal]{IEEEtran}
\usepackage{amsmath}
\usepackage{amsfonts}
\usepackage{amssymb}
\usepackage{amsthm}
\usepackage{tabularx}
\usepackage{mathrsfs} 

\usepackage{url}
\usepackage{hyperref}

\newtheorem{remark}{Remark}

\newtheorem{Proposition}{Proposition}
\newtheorem{lemma}{Lemma}
\newtheorem{corollary}{Corollary}

\usepackage{stfloats}
\usepackage{float}
\usepackage{graphicx}

\usepackage{threeparttable}

\usepackage{cite}
\usepackage{xcolor}
\usepackage{subfigure}

\makeatletter
\def\blfootnote{\xdef\@thefnmark{}\@footnotetext}
\makeatother

\begin{document}
	
		\title{\huge{Physical Layer Security Performance of \textcolor{black}{Cooperative} Dual RIS-aided V2V NOMA Communications}} 
\author{Farshad~Rostami~Ghadi, \IEEEmembership{Member},  \textit{IEEE},~Masoud~Kaveh,~Kai-Kit~Wong, \IEEEmembership{Fellow}, \textit{IEEE}, and Diego~Mart\'in}
	\maketitle
	
	\begin{abstract}
This paper investigates the performance of physical layer security (PLS) in a vehicle-to-vehicle (V2V) communication system, where a transmitter vehicle exploits a dual reconfigurable intelligent surface (RIS) to send confidential information to legitimate receiver vehicles under the non-orthogonal multiple access (NOMA) scheme in the presence of an eavesdropper vehicle. In particular, it is assumed that an RIS is near the transmitter vehicle and another RIS is close to the receiver vehicles to provide a wider smart radio environment. Besides, we suppose that the channels between two RISs suffer from the Fisher-Snedecor $\mathcal{F}$ fading model. Under this scenario, we first provide the marginal distributions of equivalent channels at the legitimate receiver vehicles by exploiting the central limit theorem (CLT). Then, in order to evaluate the PLS performance of the considered system model, we derive analytical expressions of the average secrecy capacity (ASC), secrecy outage probability (SOP), and secrecy energy efficiency (SEE) by using the Gauss-Laguerre quadrature and the Gaussian quadrature techniques. Moreover, to gain more insights into the secrecy performance, the asymptotic expression of the ASC is obtained. The numerical results indicate that incorporating the dual RIS in the secure V2V communication under the NOMA scheme can significantly provide  ultra-reliable transmission and guarantee more secure communication for intelligent transportation systems (ITS). 
	\end{abstract}
	\maketitle
	\blfootnote{\noindent 
	}
	\blfootnote{\noindent The work of F. Rostami Ghadi and K. K Wong is supported by the Engineering and Physical Sciences Research Council (EPSRC) under Grant EP/W026813/1. The work of D. Mart\'in is supported in part by the Project PRESECREL under Grant PID2021-124502OB-C43, and in part by the Ministerio de Ciencia e Investigación (Spain), in relation to the Plan Estatal de Investigación Científica y Técnica y de Innovación 2017-2020. For the purpose of open access, the authors will apply a Creative Commons Attribution (CC BY) licence to any Author Accepted Manuscript version arising.}
	\blfootnote{\noindent Farshad Rostami Ghadi and Kai-Kit Wong are with the Department of Electronic and Electrical Engineering, University College London, WC1E
		6BT London, UK. (E-mail: $\rm \{f.rostamighadi,  kai-kit.wong\}@ucl.ac.uk$).}
		\blfootnote{\noindent Masoud Kaveh is with the Department of Information and Communication Engineering, Aalto University, 02150 Espoo, Finland (e-mail: $\rm masoud.kaveh@aalto.fi$).}
	\blfootnote{\noindent Diego Mart\'in is with the Department of Computer Science, Escuela de Ingeniería Informática de Segovia, Universidad de Valladolid, Segovia, 40005, Spain (e-mail:diego.martin.andres@uva.es).
			
		\thanks{Corresponding author: Diego Mart\'in.}
	
	\thanks{Digital Object Identifier 10.1109/XXX.2021.XXXXXXX}}
	
	\section{Introduction}\label{sec-intro}
	\IEEEPARstart{T}{he} advent of the sixth generation (6G) of wireless communication systems heralds an era where the seamless interconnectivity of devices is paramount \cite{guo2022vehicular}. Among the various types of connections, vehicular communications stand out for their critical role in enhancing road safety, traffic efficiency, and the provision of infotainment services \cite{wang2021green}. Vehicle-to-everything (V2X) communications, a subset of which includes vehicle-to-vehicle (V2V) interactions, are essential components of future intelligent transportation systems (ITS). These systems utilize the exchange of information between vehicles and various entities in the traffic ecosystem to support autonomous driving and smart mobility solutions \cite{noor20226g,kaveh2020lightweight,huang2022non}.
	In the pursuit of hyper-connectivity and intelligent network solutions proposed for 6G, reconfigurable intelligent surfaces (RIS) has appeared as a transformative technology \cite{basar2019wireless}. RIS-aided communication networks can intelligently alter the propagation environment to enhance signal quality, extend coverage, and increase spectral and energy efficiency \cite{basharat2021reconfigurable}. Moreover, to broaden the utility of RISs across the complete spatial domain, the concept of simultaneously transmitting and reflecting RISs
(STAR-RISs) that possess the ability to concurrently transmit and reflect signals has been suggested.
\cite{mu2021simultaneously,palomares2023enabling,liu2021star}. 	
 As vehicular networks become more interconnected, the security of transmitted information becomes increasingly vulnerable to interception and unauthorized access \cite{lu2018survey}. 
 It has been a significant challenge for maintaining the secrecy of sensitive data.
	Physical layer security (PLS) approaches have been advocated to exploit the inherent randomness of communication channels to enhance security, providing an additional layer of protection that complements upper-layer encryption methods \cite{shiu2011physical}.
	The performance of PLS is primarily determined by the unpredictability and variability of the wireless channel. This is due to the fact that the security in PLS is achieved by ensuring that the signal quality at the legitimate receiver is significantly better than that at the potential eavesdropper. 
	Recently, it has been shown that the PLS performance is greatly enhanced by advanced technologies such as RIS \cite{zhang2021physical,vega2022physical,ghadi2022ris,kaveh2023secrecy,tang2021novel}. RIS can intelligently manipulate the propagation environment to improve the signal quality at the legitimate receiver while degrading it at the eavesdropper. This not only strengthens the security of the communication but also improves the overall spectral efficiency of the network.
	Moreover, with the evolution of multiple access schemes, non-orthogonal multiple access (NOMA) has emerged as a superior alternative to traditional orthogonal multiple acess (OMA) techniques \cite{saito2013non}. NOMA leverages the power domain for user multiplexing, allowing multiple users to share the same frequency resources but with different power levels, thereby increasing spectral efficiency and supporting massive connectivity--a key requirement in dense V2V networks \cite{di2017v2x}.
	\subsection{Related Works}
	There have been great efforts in recent years focusing on RIS-aided communications in vehicular networks for various scenarios. 
 The authors in \cite{gu2022socially} evaluated the optimization of socially aware V2X networks augmented with RIS technology. They focused on maximizing the sum capacity of vehicle-to-infrastructure (V2I) links under the constraints of reliable V2V communications, through joint optimization of power allocation, RIS reflection coefficients, and spectrum allocation. Their proposed algorithm demonstrates significant improvements in vehicular communication quality, illustrating the transformative impact of RIS in enhancing vehicular network performance. 
 In \cite{chen2021robust}, the authors investigated robust transmission strategies for RIS-aided millimeter wave (mmWave) vehicular communications. They proposed a transmission scheme that adapts to quickly changing channel state information (CSI), crucial for high-mobility scenarios. The focus on statistical CSI rather than instantaneous CSI addresses the challenge of high signaling overhead in RIS-aided systems. Their results confirmed the effectiveness of their approach, highlighting the potential of RIS in elevating mmWave vehicular communication under dynamic conditions. 
 By considering the phase error and mobility of vehicles in RIS-assisted communication systems, the author in \cite{chapala2023intelligent} employed a new scheme
	to demonstrate the distribution of the resultant channel and then derived analytical expressions of average bit-error-rate (BER) under the generalized-$\mathcal{K}$ shadowed fading distribution. The authors in  \cite{kavaiya2023restricting} focused on mitigating passive attacks in 6G vehicular networks from a PLS perspective. They presented analytical expressions of OP and secrecy outage probability (SOP), considering vehicular to infrastructure scenarios aided by RIS. Their findings suggested that RIS can enhance the secrecy performance more effectively than traditional access points, indicating its pivotal role in securing future vehicular networks against eavesdropping threats. 
 In \cite{mensi2022performance}, the authors explored the performance of partial RIS selection versus partial relay selection in vehicular communications. They derived closed-form expressions for different performance metrics, such as EC and average secrecy capacity (ASC). Their comparison showed the advantages of partial RIS selection over single RIS architecture and traditional relay systems, especially in high mobility and high quality of service (QoS) requirement scenarios of wireless vehicular networks. 
 The secrecy performance of RIS-aided vehicular communications by considering both V2V and V2I scenarios was analyzed in \cite{ai2021secure}. The authors derived closed-form expressions of the SOP, confirming the potential of RIS in enhancing secure communications in vehicular networks. Their work underscored the versatility of RIS in different vehicular communication scenarios, bolstering the security against the passive eavesdroppers. \textcolor{black}{Besides, the authors in \cite{gu2022physical} analyzed the SOP and ASC for different two scenarios in RIS-aided secure communication systems, where one is that the eavesdropper distributes close to the transmitter without the RIS orientation, and the other is that the eavesdropper locates close to the legitimate receiver and in the presence of the RIS. Moreover, by assuming two models that the full CSI and statistical CSI are available at the transmitter in a secure STAR-RIS-aided NOMA communication system, the sum-rate maximization and SOP minimization problems were evaluated in \cite{zhang2023star}. }

\textcolor{black}{Aiming to cover more users and enhance the QoS, a multi-RIS wireless system was considered in \cite{do2021multi}, where the authors derived tight approximate closed-form expressions for the OP and ergodic capacity (EC), exploiting the central limit theorem (CLT). Their results indicated the superiority of using the exhaustive RIS-aided scheme over the opportunistic RIS-aided model. 
In addition, an optimization problem was evaluated in \cite{zhao2021cooperative} to minimize the mean-squared error of the recovered data in the presence of timing offsets and phase shifts in a multi-RIS-aided communication system. 
Assuming that the direct link between the transmitter and receivers exists, the author in \cite{phan2022performance} derived the closed-form expression for symbol error probability (SEP) in a multi-RIS-aided communication system under Nakagami-$m$ fading channels, where their results showed that considering the multi-RIS is more beneficial compared with using the single-antenna system.  Additionally, in \cite{kumar2023performance}, the multi-RIS scenario was applied to the D2D communication system under the NOMA scheme, where the outage and throughput performance were theoretically analyzed over the Rician fading channels. In this regard, numerical results revealed the significant advantages of using the multi-RIS model in comparison to the single-RIS model in terms of the OP and system throughput. 
By assuming an unmanned aerial vehicle (UAV) as a transmitter, compact analytical expressions of the OP and achievable data rates (ADR) were also derived in \cite{nguyen2023performance} for a multi-RIS-aided NOMA system under Nakagami-$m$ fading channels. Moreover, based on the effect of transceiver hardware impairments, the authors in \cite{tran2022exploiting} obtained the OP, throughput, ADR, and SEP expressions of the multiple-RIS-assisted wireless systems over Nakagami-$m$ fading channels.} 

\textcolor{black}{The aforesaid contributions mainly analyzed the performance and optimization problems for single-RIS or non-cooperative multi/dual RIS-aided communication setups. However, analyzing a cooperative multi/dual RIS-aided scheme has received attention over recent years to extend the coverage region and support cell-edge users. For this purpose, \cite{ma2022cooperative} considered a cooperative multi-RIS scenario to design beamforming, solving a sum-rate optimization problem. In this model, the multiple RISs cooperate with each other in order to assist the downlink communications from a transmitter to its users. Following the proposed model in \cite{ma2022cooperative}, the authors in \cite{zhang2023double} considered a cooperative dual-RIS-aided multi-user multiple-input multiple-output (MIMO) communication system, where the mean-square-error (MSE) minimization problem by jointly optimizing the active transmit beamforming, the receive equalizer, and the passive beamforming at each RIS was investigated. 
 Moreover, a cooperative double-RIS-aided multi-user mmWave uplink communication system was considered in \cite{xue2023multi}, aiming to maximize the system throughput by jointly optimizing multi-user transmission power, active beamforming, and passive beamforming. Besides, \cite{ma2023power} considered a double cooperative-RIS-assisted uplink NOMA system with inter-RIS reflections and then formulated an optimization problem for minimizing the total transmission power for the proposed model.}
 From the performance analysis viewpoint, by examining V2I communication aided by a dual RIS, the authors in \cite{shaikh2022performance} derived analytical expressions for the OP, energy efficiency (EE), and spectral efficiency (SE) under Nakagami-$m$ fading channels. Their results indicated that considering the cooperative dual RIS instead of single RIS can significantly enhance the system performance in  ITS. 
 Moreover, by considering both NOMA and OMA schemes, the authors in \cite{ghadi2023performance} evaluated the performance of  cooperative RIS/STAR-RIS-assisted V2V communications, where they obtained closed-form expressions of the OP, EC, and EE under Fisher-Snedecor $\mathcal{F}$ fading channels, exploiting the CLT. They revealed that the NOMA approach under the consideration of RIS and STAR-RIS can provide better performance in comparison with OMA scenario for ITS. 
 \textcolor{black}{Table \ref{table_1} shows the differences between previous studies and our research, especially highlighting that this is the first work that studies the PLS performance of (cooperative) dual RIS-aided communication system.}

\begin{table*} [t]
    \centering
    \begin{threeparttable}
    \caption{\textcolor{black}{Comparison between the Related Works and Our Work} }
    \label{table_1}
    {\color{black} 
    \begin{tabular}{c|cccccccc}
    \hline \hline
      Works   & Single RIS & Multi/dual RIS & Cooperative multi/dual RIS & V2V & PLS & NOMA & Performance & Optimization \\
    \hline
    \cite{gu2022socially,chen2021robust}  & \checkmark & $\times$ & $\times$ & \checkmark & $\times$  & $\times$ & $\times$ & \checkmark \\
    \hline
      \cite{chapala2023intelligent,mensi2022performance} & \checkmark & $\times$ & $\times$ & \checkmark & $\times$  & $\times$ & \checkmark & $\times$ \\
    \hline
       \cite{kavaiya2023restricting,ai2021secure}  & \checkmark & $\times$ & $\times$ & \checkmark & \checkmark  & $\times$ & \checkmark & $\times$ \\
    \hline
        \cite{gu2022physical} & \checkmark & $\times$ & $\times$ & $\times$ & \checkmark  & $\times$ & \checkmark & $\times$  \\
    \hline
        \cite{zhang2023star} & \checkmark & $\times$ & $\times$ & \checkmark & \checkmark  & \checkmark & $\times$ & \checkmark  \\
    \hline
       \cite{do2021multi,tran2022exploiting} & $\times$ & \checkmark & $\times$ & \checkmark & $\times$  & $\times$ & \checkmark & $\times$  \\
         \hline
        \cite{zhao2021cooperative,phan2022performance} & $\times$ & \checkmark & $\times$ & $\times$ & $\times$  & $\times$ & $\times$ & \checkmark \\
        \hline
        \cite{kumar2023performance,nguyen2023performance} & $\times$ & \checkmark & $\times$ & $\times$ & $\times$  & \checkmark & $\times$ & \checkmark  \\
        \hline
        \cite{ma2022cooperative,zhang2023double,xue2023multi} & $\times$ & $\times$ & \checkmark & $\times$ & $\times$  & $\times$ & $\times$ & \checkmark \\
        \hline
        \cite{ma2023power} & $\times$ & $\times$ & \checkmark & $\times$ & $\times$  & \checkmark & $\times$ & \checkmark \\
        \hline
        \cite{shaikh2022performance} & $\times$ & $\times$ & \checkmark & \checkmark & $\times$  & $\times$ & \checkmark & $\times$ \\
        \hline
        \cite{ghadi2023performance} & $\times$ & $\times$ & \checkmark & \checkmark & $\times$  & \checkmark & \checkmark & $\times$ \\
    \hline
        Ours & $\times$ & $\times$ & \checkmark & \checkmark & \checkmark  & \checkmark & \checkmark & $\times$  \\
    \hline
    \end{tabular}
    } 
    \begin{tablenotes}
        \item \textcolor{black}{ \checkmark: Item is supported, $\times$: Item is not supported.}
    \end{tablenotes}
\end{threeparttable}
\end{table*}

	\subsection{\textcolor{black}{Motivation and Contributions}}
	\textcolor{black}{It is widely known one of the key reasons of using the RIS is to extend the coverage region, especially in a case where the direct links between the transmitter and receivers are blocked due to obstacles. In this case, a single RIS-aided communication system can be beneficial to support those receivers. However, especially in urban dynamic ultra-dense cellular networks such as V2V, V2I, and device-to-device (D2D) communications,  mobile users or vehicles can be located very far from the transmitter or even the RIS so that the links between the corresponding RIS and them (cell-edge users/vehicles) are blocked and the RIS cannot cover those users. In such a likely practical scenario, considering additional RIS (two or more) to serve the far users can be beneficial since each RIS can cooperate with its previous nearest RIS to extend the network coverage area. This model mainly referred as \textit{cooperative} multi/dual RIS-aided communication systems.} Motivated by the significant impact of using \textcolor{black}{cooperative} RISs in extending coverage area and providing high spectral/energy efficiency over the next generation of intelligent V2V communications, the remarkable advantages of exploiting PLS to guarantee secure and reliable transmission, and the key role of Fisher-Snedecor $\mathcal{F}$ distribution in precise modeling and characterization of the simultaneous occurrence of multipath fading and shadowing, in this paper, we consider a \textcolor{black}{cooperative} dual RIS-aided V2V secure communication under the NOMA transmission scheme, where the fading channels is modeled as Fisher-Snedecor $\mathcal{F}$ distribution. \textcolor{black}{It is worth mentioning that this proposed model is totally different from the multi-RIS scenario defined in the literature. In the most previous multi-RIS communication systems, the RISs did not cooperate with each other to extend the coverage area; instead, a large number of RISs were deployed near each group of users so that each RIS could separately support those associated groups. In other words, for each group of users, one RIS is required, where it’s like the case that many single-RIS without cooperation are considered in the system model. Hence, in this case, the coverage region will not be extended in comparison with the single-RIS or cooperative dual-RIS models, and only more groups of users can be supported thanks to increasing the number of single-RIS. Therefore,} to the best of the author's knowledge, no prior studies have addressed the \textcolor{black}{PLS over cooperative dual RISs-aided vehicular communications under NOMA scheme.} 
 In particular, we assume a transmitter vehicle aims to send a confidential message to legitimate receivers by employing two RISs, while an eavesdropper vehicle tries to decode the information. 
 Moreover, we assume that the RISs are close to the legitimate and eavesdropper vehicles so that the corresponding links can be represented as deterministic line-of-sight (LoS) channels. 
 In this regard, we first derive the marginal distributions of the equivalent channels at legitimate receiver and eavesdropper vehicles, and then we develop analytical formulations for key secrecy performance metrics to assess the secrecy performance of the proposed system model. Hence, the key contributions of our study can be outlined as follows
	
	\textbullet\ \,First, we provide the marginal distributions of the signal-to-noise ratio (SNR) and signal-to-interference-plus-noise ratio (SINR) at the legitimate receiver and eavesdropper vehicles with the help of the CLT.
	
	\textbullet\ \,Next, we derive the analytical expressions of the ASC, SOP, and secrecy energy efficiency (SEE) by utilizing the Gauss-Laguerre quadrature and the Gaussian quadrature methods.
	
	\textbullet\ \,Then, to acquire deeper understanding of the performance of the obtained secrecy metrics, we derive the asymptotic expressions of the ASC in the high SNR regime.
	
	\textbullet\ \,Finally, we evaluate the secrecy performance of the considered system model in terms of the ASC, SOP, and SEE. To this purpose, we validate the validation of the analytical expressions with the Monte-Carlo simulation; namely, the numerical findings validate the effectiveness of incorporating dual RIS in V2V communication can significantly provide a more secure and reliable transmission for ITS.
	
	\subsection{Organization }
The reminder for this paper is organized as follows: Section \ref{sec-sys} describes the channel model and the statistical characterization. In Section \ref{sec-per}, the secrecy performance metrics are analyzed. The numerical results are also presented in Section \ref{sec-num}, and the conclusions are provided in Section \ref{sec-con}.
\section{System Model}\label{sec-sys}
 \subsection{Channel Model}
 We examine a secure wireless V2V communication within the framework of NOMA as illustrated in Fig. \ref{fig:fig1}, where a transmitter $v_\mathrm{t}$ aims to send a confidential signal \footnote{\textcolor{black}{In V2V communication systems, the types of confidential messages that usually require to be transmitted over a wiretap channel can vary depending on the specific application and context. However, confidential information such as safety-critical messages, traffic conditions, cooperative driving strategies, privacy-preserving location sharing, lane change and merging intentions, etc., are the most critical common information that needs to be protected from eavesdropping to prevent malicious manipulation or interference.}} with the help of two RISs to the legitimate receiver vehicles $v_{\mathrm{r}_i}$, $i\in\{1,2\}$, whereas an eavesdropper vehicle $v_\mathrm{e}$ attempts to decode the desired information \footnote{\textcolor{black}{In this scenario, we assume that the vehicles move at nearly constant speed or low speed (e.g., traffic conditions), which leads the constant Doppler frequency. By doing so, the effect of vehicle speed or the user mobility speed is mainly ignored
\cite{ai2021secure,shaikh2022performance}}}. We assume that the direct links between $v_\mathrm{t}$ and both $v_{\mathrm{r}_i}$ and $v_\mathrm{e}$ are blocked due to obstacles such as buildings, etc. Hence, we assume that the first RIS with $M_1$ elements is located near the transmitter vehicle $v_\mathrm{t}$ and the second RIS with $M_2$ elements is deployed close to the legitimate receivers $v_{\mathrm{r}_i}$ and eavesdropper $v_\mathrm{e}$. We also consider that the distance between $v_\mathrm{t}$ and the first RIS, $d_{v_\mathrm{t}R_1}$, as well as the distance between the second RIS and $v_{\mathrm{r}_i}$, $d_{R_2v_{\mathrm{r}_i}}$, and $v_\mathrm{e}$, $d_{R_2v_\mathrm{e}}$, are small, and thus, the respective channels can be appropriately modeled as determinist line-of-sight (LoS) channels. Meanwhile, we assume that the distance between two RISs, $d_\mathrm{R_1R_2}$, is large enough; thereby, the corresponding links between the $k$-th elements of the first RIS and the $l$-th element of the second RIS are quasi-static fading channels. According to the NOMA scheme, let the transmitter vehicle $v_\mathrm{t}$ exploits superposition coding, i.e., to simultaneously transmit the corresponding signals to $v_{\mathrm{r}_i}$ over the power domain. Specifically, without loss of generality, we consider $v_\mathrm{r_1}$ as the near receiver vehicle (i.e., strong receiver vehicle) and $v_\mathrm{r_2}$ as the far receiver vehicle (i.e., weak receiver vehicle). By doing so, more power is always allocated to the weak receiver vehicle $v_\mathrm{r_2}$, i.e., $p_\mathrm{r_2}\ge p_\mathrm{r_1}$, due to the user fairness, where $p_{\mathrm{r}_i}$ is the power allocation factor that satisfies $p_\mathrm{r_1}+p_\mathrm{r_2}=1$. Therefore, the received signal at the vehicle $\mathrm{n}\in\{\mathrm{r}_i,\mathrm{e}\}$ can be expressed as
 \begin{align}
 	y_\mathrm{n}=\sqrt{P}\left(\sum_{k=1}^{M_1}\sum_{l=1}^{M_2}\mathrm{e}^{j\phi_k}h_{k,l}\mathrm{e}^{j\psi_{\mathrm{n},l}}\hspace{-1mm}\right)\hspace{-1.5mm}\left(\sqrt{p_\mathrm{r_1}}x_\mathrm{r_1}+\sqrt{p_\mathrm{r_2}}x_\mathrm{r_2}\right)\hspace{-1mm}+z_\mathrm{n}\label{},
 \end{align}
 where $P$ is the total transmit power, $x_{\mathrm{r}_i}$ denotes the symbol sent by the transmitter vehicle $v_\mathrm{t}$ to $v_\mathrm{n}$, and $z_\mathrm{n}\sim\mathcal{CN}\left(0,\delta_\mathrm{n}^2\right)$ denotes the additive white Gaussian noise (AWGN) at the vehicle $v_\mathrm{n}$ that is modeled as a zero-mean complex Gaussian distribution with variance $\delta_\mathrm{n}^2$. Besides, the fading channel from the first RIS to the second RIS is defined by $h_{k,l}=\eta_{k,l}d_{\mathrm{R_1R_2}}^{-\kappa}\mathrm{e}^{-j\zeta_{k,l}}$, where $\eta_{k,l}$ and $\zeta_{k,l}$ denote the amplitude and the phase of $h_{k,l}$, respectively. The term $\kappa>2$ is also the path-loss exponent. Moreover, $\phi_k$ and $\psi_{\mathrm{n},l}$ are the variable phases generated by
 the $k$-th element of the first RIS and the $l$-th element of the second RIS \footnote{\textcolor{black}{In this work, without loss of generality, we assume that an ideal, sophisticated synchronization algorithm that dynamically adjusts the phase shifting of each RIS in response to real-time changes in channel conditions is exploited, ensuring that synchronization is maintained even amidst rapid fluctuations \cite{huang2019reconfigurable}.}}, respectively. Additionally, regarding the assumption of LoS between vehicles and RISs, the corresponding path-loss for a link distance $d$ at the carrier frequency of $3$ GHz can be expressed as \cite{bjornson2019intelligent}
 \begin{align}
 	\mathcal{D}\left(d\right)\left[\mathrm{dB}\right]=-37.5-22\log_{10}\left(d/1\mathrm{m}\right).
 \end{align}
  \begin{figure} [t]
 	\centering
 	\includegraphics[width=0.47\textwidth]{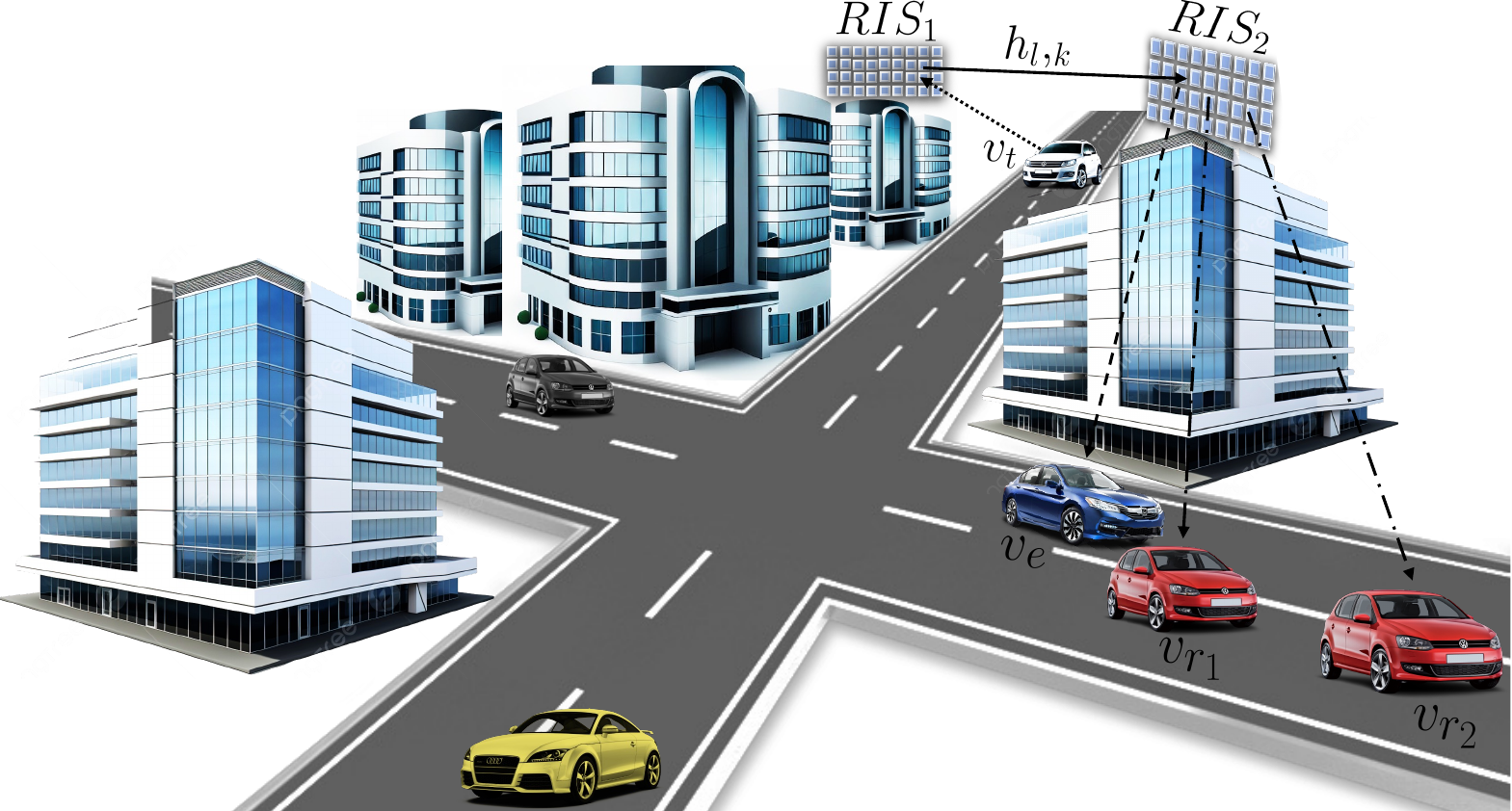}
 	\caption{System model indicates a secure dual RIS-aided V2V NOMA communication.}
 	\label{fig:fig1}
 \end{figure}
 \subsection{Instantaneous SNR/SINR}
 Regarding the NOMA scheme, the successive interference cancellation (SIC) is conducted at the strong receiver vehicle $v_\mathrm{r_1}$ to gain high rate performance, whereas the weak receiver vehicle $v_\mathrm{r_2}$ decodes its signal directly by considering interference as noise Hence, we can define the SINR of the SIC process as
 \begin{align}
 	\gamma_\mathrm{sic}=\frac{\bar{\gamma}p_{\mathrm{r}_2}\mathcal{D}d_{\mathrm{R_1R_2}}^{-\kappa}\left|\sum_{k=1}^{M_1}\sum_{l=1}^{M_2}\eta_{k,l}\mathrm{e}^{j\left(\phi_{k}+\psi_{\mathrm{r}_1,l}-\zeta_{k,l}\right)}\right|^2}{\bar{\gamma}p_{\mathrm{r}_1}\mathcal{D}d_{\mathrm{R_1R_2}}^{-\kappa}\left|\sum_{k=1}^{M_1}\sum_{l=1}^{M_2}\eta_{k,l}\mathrm{e}^{j\left(\phi_{k}+\psi_{\mathrm{r}_1,l}-\zeta_{k,l}\right)}\right|^2+1},\label{eq-sic}
 \end{align}
 in which $\bar{\gamma}=\frac{P}{\delta^2_\mathrm{r}}$ is the average transmit SNR. Next, after successful SIC, $v_\mathrm{r_1}$ eliminates the message of $v_\mathrm{r_2}$ from its received signal and deciphers its necessary information with the subsequent SNR.
 \begin{align}
 	\gamma_\mathrm{r_1}=\bar{\gamma}p_\mathrm{r_1}\mathcal{D}d_{\mathrm{R_1R_2}}^{-\kappa}\left|\sum_{k=1}^{M_1}\sum_{l=1}^{M_2}\eta_{k,l}\mathrm{e}^{j\left(\phi_{k}+\psi_{\mathrm{r_1},l}-\zeta_{k,l}\right)}\right|^2.\label{eq-r1}
 \end{align}
Meanwhile, the weak receiver vehicle $v_{\mathrm{r}_2}$ decodes its own signal and it has no power to remove the signal of $v_\mathrm{r_1}$ from the combined transmitted message. Hence, by considering the signal of $v_\mathrm{r_1}$ as extra interference, the corresponding SINR at $v_\mathrm{r_2}$ can be expressed as 
  \begin{align}
  	\gamma_\mathrm{r_2}=\frac{\bar{\gamma}p_\mathrm{r_2}\mathcal{D}d_{\mathrm{R_1R_2}}^{-\kappa}\left|\sum_{k=1}^{M_1}\sum_{l=1}^{M_2}\eta_{k,l}\mathrm{e}^{j\left(\phi_{k}+\psi_{\mathrm{r}_2,l}-\zeta_{k,l}\right)}\right|^2}{\bar{\gamma}p_\mathrm{r_1}\mathcal{D}d_{\mathrm{R_1R_2}}^{-\kappa}\left|\sum_{k=1}^{M_1}\sum_{l=1}^{M_2}\eta_{k,l}\mathrm{e}^{j\left(\phi_{k}+\psi_{\mathrm{r_2},l}-\zeta_{k,l}\right)}\right|^2+1}.\label{eq-r2}
  \end{align}
  On the other hand, by considering the multiuser detection ability of the eavesdropper, we assume that it can exploit the parallel interference cancellation approach to intercept the different vehicles' signals. Therefore, the received SNR at $v_\mathrm{e}$ can be formulated as  
   \begin{align}
  	\gamma_\mathrm{e}=\bar{\gamma}_\mathrm{e}p_\mathrm{r_1}\mathcal{D}d_{\mathrm{R_1R_2}}^{-\kappa}\left|\sum_{k=1}^{M_1}\sum_{l=1}^{M_2}\eta_{k,l}\mathrm{e}^{j\left(\phi_{k}+\psi_{\mathrm{r_1},l}-\zeta_{k,l}\right)}\right|^2,\label{}
  \end{align}
in which $\bar{\gamma}_\mathrm{e}=\frac{P}{\delta_\mathrm{e}^2}$ is the average SNR at $v_\mathrm{e}$.
  \subsection{Statistical Characterization}
  In order to evaluate the performance of secrecy metrics such as SOP and ASC, we first need to characterize the distributions of $\gamma_{{\mathrm{r}}_i}$ and $\gamma_{\mathrm{e}}$. 
  
  \subsubsection{Distribution of $\gamma_{{\mathrm{r}}_i}$} In order to optimize the SNR and SINR at the receiver vehicles $v_{\mathrm{r}_i}$, we consider that the CSI of $v_{\mathrm{r}_i}$ is available at the RISs, thereby, the RISs can set an ideal phase shifting, i.e., $\zeta_{k,l}=\phi_k+\psi_{\mathrm{r}_i,l}$, in \eqref{eq-sic}-\eqref{eq-r2}. Hence, regarding the CLT when $M_1, M_2\gg 1$, the random variable $U=\left|\sum_{k=1}^{M_1}\sum_{l=1}^{M_2}\eta_{k,l}\right|$ is approximated as a Gaussian random variable with mean $\mu_U=\sum_{k=1}^{M_1}\sum_{l=1}^{M_2}\mu_{k,l}$ and variance  $\sigma^2_U=\sum_{k=1}^{M_1}\sum_{l=1}^{M_2}\sigma^2_{k,l}$. Given that we consider the legitimate channels following the Fisher-Snedecor $\mathcal{F}$ distribution\footnote{\textcolor{black}{The Fisher-Snedecor $\mathcal{F}$ model is able to accurately capture the main fading distribution models such as Rayleigh (i.e., $m_1 = 1, m_2\rightarrow \infty$) and Nakagami-m (i.e., $m_1 = m, m_2\rightarrow\infty$) distributions. Besides, experimental results showed the superiority of the Fisher-Snedecor $\mathcal{F}$ fading distribution rather than the than the Generalized-$\mathcal{K}$ fading model, especially in D2D communications, where the mobility of users is momentous \cite{yoo2017fisher}.}}, $\mu_{k,l}=\frac{m_2}{m_2-1}$ and $\sigma^2_{k,l}=\frac{m_2^2\left(m_1+m_2-1\right)}{m_1\left(m_2-1\right)^2\left(m_2-2\right)}$ can be defined as the mean and variance of the random variable $\eta_{k,l}$ for all $k=\left\{1,\dots,M_1\right\}$ and $l=\left\{1,\dots,M_2\right\}$, respectively. The terms $m_1$ and $m_2$ are fading parameters that represent the fading severity and the amount of shadowing of the root-mean-square (rms) signal power, respectively. Therefore, the corresponding CDF for the random variable $U$ is defined as
  \begin{align}
  	F_U\left(u\right)=\frac{1}{2}\left[1+\mathrm{erf}\left(\frac{u-\mu_U}{\sigma_u^2\sqrt{2}}\right)\right], \label{eq-cdf-u}
  \end{align}
  where $\mathrm{erf}(s)=\frac{2}{\sqrt{\pi}}\int_0^s\mathrm{e}^{-t^2}dt$ is the error function. Additionally, the mean $\mu_U$ and variance $\sigma^2_\mu$ under Fisher-Snedecor $\mathcal{F}$ can be expressed respectively as
  \begin{align}
	\mu_U=\frac{M_1M_2m_2}{m_2-1}, \label{eq-mean}
\end{align}
\begin{align}
	\sigma^2_U=\frac{M_1M_2m_2^2\left(m_1+m_2-1\right)}{m_1\left(m_2-1\right)^2\left(m_2-2\right)}. \label{eq-var}
\end{align}
   Then, by applying the transformation of random variables into \eqref{eq-cdf-u}, the CDF of $\gamma_{{\mathrm{r}}_1}$ and $\gamma_{{\mathrm{r}}_2}$ are respectively given by 
      \begin{align}
  	F_{\gamma_{\mathrm{r}_1}}\left(\gamma_{\mathrm{r}_1}\right)=\frac{1}{2}\left[1+\mathrm{erf}\left(\frac{\sqrt{\gamma_{\mathrm{r}_1}}-\mu_U\sqrt{\bar{\gamma}A_1}}{\sigma_U^2\sqrt{2\bar{\gamma}A_1}}\right)\right],
  \end{align}
   \begin{align}
  	F_{\mathrm{\gamma_{\mathrm{r}_2}}}\left(\gamma_{\mathrm{r}_2}\right)= \frac{1}{2}\left[1+\mathrm{erf}\left(\frac{\sqrt{\frac{\gamma_{\mathrm{r}_2}}{\bar{\gamma}A_2-\bar{\gamma}A_1\gamma_{\mathrm{r}_2}}}-\mu_U}{\sigma_U^2\sqrt{2}}\right)\right], 
  \end{align}
  in which $A_i=p_{\mathrm{r_i}}\mathcal{D}d_{\mathrm{R_1R_2}}^{-\kappa}$ for $i\in\{1,2\}$.
    \subsubsection{Distribution of $\gamma_{{\mathrm{e}}}$} By optimizing the phase shifts of the RIS elements according to the conditions of the legitimate link, the phase shifts experienced by the eavesdropper's link can end up being uniformly distributed. Moreover, we assume that the CSI of $v_\mathrm{e}$ is not known and the RISs cannot maximize the SNR of the eavesdropper. Then, according to the CLT, $\gamma_{{\mathrm{e}}}$ can be estimated by an exponential random variable. Thus, the PDF and CDF of $\gamma_{{\mathrm{e}}}$ can be respectively defined as \cite{sanchez2020physical}
        \begin{align}
    	f_{\gamma_{\mathrm{e}}}\left(\gamma_{\mathrm{e}}\right)=\lambda_{\mathrm{e}}\mathrm{e}^{-\lambda_\mathrm{e}\gamma_{\mathrm{e}}},
    \end{align}
    \begin{align}
    	F_{\gamma_{\mathrm{e}}}\left(\gamma_{\mathrm{e}}\right)=1-\mathrm{e}^{-\lambda_\mathrm{e}\gamma_{\mathrm{e}}},\label{eq-cdf-e}
    \end{align}
    in which $\lambda_{\mathrm{e}}=\frac{1}{\bar{\gamma}_\mathrm{e}A_1\mu_u }$.
\section{Secrecy Metrics Analysis} \label{sec-per}
\textcolor{black}{In V2V secure communication systems, important secrecy metrics such as ASC, SOP, and SEE need to be analyzed in various practical applications. Generally speaking, such secrecy metrics can be utilized in different scenarios, like traffic management, safety applications, privacy protection, etc., to achieve insights for designing V2V networks.} \textcolor{black}{Specially, in V2V secure communication systems, safety-related messages, such as collision warnings or emergency brake alerts, traffic flow information such as congestion alerts and route recommendations, and personal or sensitive information such as location data and driver behavior patterns, need to be transmitted securely to prevent tampering or spoofing by eavesdroppers. Therefore, in this section}, we first derive analytical expressions of the ASC, SOP, and SEE, by exploiting the obtained CDF and PDF. Then, we derive the asymptotic expressions of the ASC in the high SNR conditions. 
  \subsection{ASC Analysis}
  One of the main goals of PLS is to guarantee secure communication between transmitters and receivers. Therefore, in order to achieve such secure transmission in the considered dual RIS-aided NOMA V2V system, the data rates of the confidential messages $x_{\mathrm{r}_i}$ for receiver vehicles $v_{\mathrm{r}_i}$ need to be set based on the channel conditions and decoding process at the respective receiver and eavesdropper vehicles. To this end, secrecy capacity is defined as the difference between the channel capacities corresponding to legitimate and wiretap links. In particular, the secrecy capacity for the transmission of information from transmitter vehicle $v_{\mathrm{t}}$ to receiver vehicle $v_{\mathrm{r}_i}$ is mathematically given by
  \begin{align}
  	C_{\mathrm{s}}^{\mathrm{r}_i}\left(\gamma_{{\mathrm{r}}_i},\gamma_{\mathrm{e}}\right)=\left[C_{\mathrm{r}_i}-C_\mathrm{e}\right]^+, \quad i\in\left\{1,2\right\},
  \end{align}
  where \textcolor{black}{$\left[x\right]^+=\max\left\{x,0\right\}$ and}  $C_\mathrm{n}=\log_2\left(1+\gamma_{\mathrm{n}}\right)$.
  Then, given that $C_\mathrm{s}^{\mathrm{r}_i}$ is a random function of $\gamma_{{\mathrm{r}}_i}$ and $\gamma_{{\mathrm{e}}}$, we can express the ASC as follow
  \begin{align}
  	\bar{C}_\mathrm{s}^{\mathrm{r}_i}\left[\mathrm{bps/Hz}\right]\triangleq \mathbb{E}\left[C_\mathrm{s}^{\mathrm{r}_i}\left(\gamma_{{\mathrm{r}}_i},\gamma_{{\mathrm{e}}}\right)\right]. \label{eq-asc}
  	\end{align}
  \begin{Proposition}\label{pro-asc1}
  	The ASC for the receiver vehicle $v_\mathrm{r_1}$ for the studied secure dual RIS-aided V2V NOMA communication system is given by
  	\begin{align}
  		\bar{C}_\mathrm{s}^\mathrm{r_1}\approx\frac{2}{\ln 2} \sum_{\tilde{n}=1}^{\tilde{N}}w_{\tilde{n}} \frac{\left(1-\mathrm{e}^{-\lambda_\mathrm{e}\epsilon_{\tilde{n}}}\right)}{\mathrm{e}^{\epsilon_{\tilde{n}}}\left(1+\epsilon_{\tilde{n}}\right)}\mathrm{erfc}\left(\frac{\sqrt{\epsilon_{\tilde{n}}}-\mu_U\sqrt{\bar{\gamma}A_1}}{\sigma_U^2\sqrt{2\bar{\gamma}A_1}}\right),\label{eq-asc-v1}
  	\end{align}
  		in which
  	\begin{align}\label{eq-w} w_{\tilde{n}}=\frac{\epsilon_{\tilde{n}}}{2\left(\tilde{N}+1\right)^2L^2_{\tilde{N}+1}\left(\epsilon_{\tilde{n}}\right)},
  	\end{align}
  	$\epsilon_{\tilde{n}}$ is the $\tilde{n}$-th root of Laguerre polynomial $L_{\tilde{N}}\left(\epsilon_{\tilde{n}}\right)$, $\tilde{N}$ defines the parameter to ensure a complexity-accuracy trade-off, $\mathrm{erfc}(s)=1-\mathrm{erf}(s)$, $A_i=p_{\mathrm{r_i}}\mathcal{D}d_{\mathrm{R_1R_2}}^{-\kappa}$ for $i\in\{1,2\}$, $\lambda_{\mathrm{e}}=\frac{1}{\bar{\gamma}_\mathrm{e}A_1\mu_u }$, and $\mu_U$ and $\sigma^2_U$ are denoted in \eqref{eq-mean}, and \eqref{eq-var}, respectively. 
  \end{Proposition}
  \begin{proof}
  The proof details are provided in Appendix \ref{app-asc1}.
  	\end{proof}
  \begin{Proposition}\label{pro-asc2}
	The ASC for the receiver vehicle $v_\mathrm{r_2}$ for the studied secure dual RIS-aided V2V NOMA communication system is given by
	\begin{align}\notag
		\bar{C}_\mathrm{s}^\mathrm{r_2}\approx\,& \frac{p_\mathrm{r_2}}{\ln 2}\sum_{\tilde{n}=1}^{\tilde{N}}q_{\tilde{n}} \frac{\left(1-\mathrm{e}^{-\frac{p_\mathrm{r_2}\lambda_\mathrm{e}}{p_\mathrm{r_1}}\left(\chi_{\tilde{n}}+1\right)}\right)}{p_\mathrm{r_1}+p_\mathrm{r_2}\left(\chi_{\tilde{n}}+1\right)}\\
		&\times\mathrm{erfc}\left(\frac{\sqrt{\frac{p_\mathrm{r_2}\left(\chi_{\tilde{n}}+1\right)}{\bar{\gamma}p_\mathrm{r_1}A_2-\bar{\gamma}A_1p_\mathrm{r_2}\left(\chi_{\tilde{n}}+1\right)}}-\mu_U}{\sigma_U^2\sqrt{2}}\right),\label{eq-asc-v2}
	\end{align}
	in which 
	\begin{align} q_{\tilde{n}}=\frac{2}{\left(1-\chi_{\tilde{n}}^2\right)\left[\mathcal{P}'_{\tilde{n}}\left(\chi_{\tilde{n}}\right)\right]^2},
		\end{align}
			$\chi_{\tilde{n}}$ is the $\tilde{n}$-th root of  Legendre polynomial $\mathcal{P}_{\tilde{n}}\left(\chi_{\tilde{n}}\right)$ and $\tilde{N}$ defines the parameter to ensure a complexity-accuracy trade-off.
\end{Proposition}
  \begin{proof}
	The proof details are provided in \ref{app-asc2}.
\end{proof}
In order to obtain a deeper understanding of the  performance of the ASC, we evaluate the asymptotic behavior of the ASC at the high SNR conditions, i.e., $\bar{\gamma}\rightarrow\infty$, in the following corollary. 
\begin{corollary}\label{cor-asc1}
The asymptotic ASC for the receiver vehicles $v_\mathrm{r_1}$ and $v_\mathrm{r_2}$ for the studied secure dual RIS-aided V2V NOMA communication system is respectively given by 
\begin{align}\label{eq-asy-c1}
\bar{C}_{\mathrm{s}, \infty}^{\mathrm{r}_1}\approx\log_2\left(\bar{\gamma}A_1\mu_{U^2}\right)+\mathrm{e}^{\lambda_{\mathrm{e}}}\mathrm{E}_1\left(\lambda_{\mathrm{e}}\right),
\end{align}
	\begin{align}\label{eq-asy-c2}
	\bar{C}_{\mathrm{s}, \infty}^{\mathrm{r}_2}\approx\log_2\left(1+\frac{p_\mathrm{r_2}}{p_\mathrm{r_1}}\right)+\mathrm{e}^{\lambda_{\mathrm{e}}}\mathrm{E}_1\left(\lambda_{\mathrm{e}}\right),
\end{align}
where $\mathrm{E}\left(s\right)=\int_s^\infty \frac{\mathrm{e}^{-t}}{t}dt$ is the exponential integral and 
\begin{align}
\mu_{U^2}=\left[\frac{M_1M_2m_2^2\left(m_1+m_2-1\right)}{m_1\left(m_2-1\right)^2\left(m_2-2\right)}+\left(\frac{M_1M_2m_2}{m_2-1}\right)^2\right].\label{eq-mu2}
\end{align}
\end{corollary}
\begin{proof}
The proof details are provided in Appendix \ref{app-asc-asy1}.
\end{proof}
\textcolor{black}{\begin{remark}
It is worth mentioning that the asymptotic expressions in \eqref{eq-asy-c1} and \eqref{eq-asy-c2} reveal intuitive insights into the ASC performance in terms of average SNR $\bar{\gamma}$, thanks to their simple structures. Given the logarithm function $\log_2\left(\bar{\gamma}\right)$ properties, it is mathematically understandable in \eqref{eq-asy-c1} that as $\bar{\gamma}$ grows for other fixed channel parameters, the asymptotic ASC of $v_\mathrm{r_1}$ increases. Besides, it can be seen in \eqref{eq-asy-c2} that the ASC for weak receiver vehicle $v_\mathrm{r_2}$ converges to a constant value in the high SNR regime $\bar{\gamma}\rightarrow\infty$. These behaviors are in alignment with the use of SIC in the strong NOMA receiver vehicle $v_\mathrm{r_1}$, which provides a larger SINR value.
for $v_\mathrm{r_1}$ compared with the weak vehicle $v_\mathrm{r_2}$. In Section \ref{sec-num}, such insights will be shown through numerical and simulation analysis to validate the accuracy of theoretical expressions.
\end{remark}}
\begin{Proposition}
The ASC for the legitimate receiver vehicle pair for the studied secure V2V NOMA communications systems is given by
\begin{align}
\bar{C}_\mathrm{s}= \bar{C}_\mathrm{s}^{\mathrm{r}_1}+\bar{C}_\mathrm{s}^{\mathrm{r}_2}.
\end{align}
\end{Proposition}
\subsection{SOP Analysis}
The SOP is an important metric to analyze the performance of secure communications, which is defined as the probability that the random secrecy capacity $C_{\mathrm{s}}^{\mathrm{r}_i}$ drops beneath a specified secrecy rate threshold $R_\mathrm{s}$. Therefore, for the considered RIS-aided V2V NOMA system, the SOP of receiver vehicles $v_{\mathrm{r}_i}$ can be defined as
\begin{align}
	P_\mathrm{sop}^{\mathrm{r}_i}=\Pr \left(C_{\mathrm{s}}^{\mathrm{r}_i}\leq R_\mathrm{s}\right).\label{eq-sop}
\end{align} 
\begin{Proposition}\label{pro-sop1}
	The SOP for the receiver vehicle $v_\mathrm{r_1}$ for the studied secure dual RIS-aided V2V NOMA communication system is given by
	\begin{align}
		P_\mathrm{sop}^\mathrm{r_1}\approx\sum_{\tilde{n}=1}^{\tilde{N}}\frac{w_{\tilde{n}}}{2} \left[1+\mathrm{erf}\left(\frac{\sqrt{\frac{R_\mathrm{th}}{\lambda_{\mathrm{e}}}\epsilon_{\tilde{n}}+\bar{R}_\mathrm{th}}-\sqrt{\bar{\gamma}A_1}\mu_U}{\sigma_U^2\sqrt{2\bar{\gamma}A_1}}\right)\right],\label{eq-sop-v1}
	\end{align}
	in which $R_\mathrm{th}=2^{R_\mathrm{s}}$, $\bar{R}_\mathrm{th}=R_\mathrm{th}-1$, and $w_{\tilde{n}}$ is defined in \eqref{eq-w}.
\end{Proposition}
\begin{proof}
	The proof details are provided in Appendix \ref{app-sop1}.
\end{proof}

\begin{Proposition}\label{pro-sop2}
	The SOP for the receiver vehicle $v_\mathrm{r_2}$ for the studied secure dual RIS-aided V2V NOMA communication system is given by
	\begin{align}
		P_\mathrm{sop}^\mathrm{r_2}\approx\sum_{\tilde{n}=1}^{\tilde{N}}\hspace{-1mm}\frac{w_{\tilde{n}}}{2}\hspace{-1mm}\left[1+\mathrm{erf}\left(\frac{\sqrt{\frac{\frac{R_\mathrm{th}}{\lambda_{\mathrm{e}}}\epsilon_{\tilde{n}}+\bar{R}_\mathrm{th}}{\bar{\gamma}A_2-\bar{\gamma}A_1\left(\frac{R_\mathrm{th}}{\lambda_{\mathrm{e}}}\epsilon_{\tilde{n}}+\bar{R}_\mathrm{th}\right)}}-\mu_U}{\sigma_U^2\sqrt{2}}\right)\right], \label{eq-sop-v2}
	\end{align}
	in which $R_\mathrm{th}=2^{R_\mathrm{s}}$, $\bar{R}_\mathrm{th}=R_\mathrm{th}-1$, and $w_{\tilde{n}}$ is defined in \eqref{eq-w}.
\end{Proposition}
\begin{proof}
	The proof details are provided in Appendix \ref{app-sop2}.
\end{proof}
\subsection{SEE Analysis}
SEE is a momentous metric to assess the effectiveness of wireless secure communication systems under energy constraints, which is defined as the ratio of the ASC, representing the reliable and secure transmission of confidential messages, to the total power consumption. In other words, SEE provides a measure of how a communication system with security considerations efficiently uses energy to guarantee a secure transmission. Therefore, for the considered dual RIS-aided system, SEE can be mathematically defined as
\begin{align}
	\mathcal{E}_\mathrm{s}=\frac{\bar{C}_\mathrm{s}}{P_\mathrm{tot}}=\frac{\bar{C}_\mathrm{s}}{P/\alpha+M_1P_\mathrm{R_1}+M_2P_\mathrm{R_2}+P_\mathrm{t_1}+P_\mathrm{t_2}},\label{eq-ee-z}
\end{align}
in which $\alpha$ is the drain efficiency of high-power amplifier (HPA), $P_{v_{\mathrm{r}_i}}$ denotes the circuit power consumption at receiver vehicle $v_{\mathrm{r}_i}$, and $P_\mathrm{R_i}$ is the power consumed by each element of the RIS $i$.
\begin{Proposition}
The SEE for the legitimate receiver vehicle pair for the studied secure V2V NOMA communications systems is given by \eqref{eq-see}.
\begin{figure*}[t]
	\normalsize
	\setcounter{equation}{27}
\begin{align}
\mathcal{E}_\mathrm{s}= \frac{\sum_{\tilde{n}=1}^{\tilde{N}} \frac{q_{\tilde{n}}p_\mathrm{r_2}\left(1-\mathrm{e}^{-\frac{p_\mathrm{r_2}\lambda_\mathrm{e}}{p_\mathrm{r_1}}\left(\chi_{\tilde{n}}+1\right)}\right)}{p_\mathrm{r_1}+p_\mathrm{r_2}\left(\chi_{\tilde{n}}+1\right)}\mathrm{erfc}\left(\frac{\sqrt{\frac{p_\mathrm{r_2}\left(\chi_{\tilde{n}}+1\right)}{\bar{\gamma}p_\mathrm{r_1}A_2-\bar{\gamma}A_1p_\mathrm{r_2}\left(\chi_{\tilde{n}}+1\right)}}-\mu_U}{\sigma_U^2\sqrt{2}}\right)+ \frac{2w_{\tilde{n}}\left(1-\mathrm{e}^{-\lambda_\mathrm{e}\epsilon_{\tilde{n}}}\right)}{\mathrm{e}^{\epsilon_{\tilde{n}}}\left(1+\epsilon_{\tilde{n}}\right)}\mathrm{erfc}\left(\frac{\sqrt{\epsilon_{\tilde{n}}}-\mu_U\sqrt{\bar{\gamma}A_1}}{\sigma_U^2\sqrt{2\bar{\gamma}A_1}}\right)}{\ln 2\left(P/\alpha+M_1P_\mathrm{R_1}+M_2P_\mathrm{R_2}+P_\mathrm{t_2}+P_\mathrm{t_1}\right)}. \label{eq-see}
\end{align}
	\hrulefill
	\vspace{-15pt}
\end{figure*}
\end{Proposition}
\section{Numerical Results}\label{sec-num}
Here, we assess the system performance in terms of the ASC, SOP, and SEE,  presenting numerical results which are double-checked by the Monte Carlo simulation. To this end, we configure the simulation parameters as $d_{v_\mathrm{t}R_1}=d_{R_2v_\mathrm{r_1}}=d_{R_2v_\mathrm{e}}=5$m, $d_{R_2v_\mathrm{r_2}}=15$m, $d_{R_1R_2}=100$m, $\delta^2_\mathrm{r}=-70$dBm, $\delta^2_\mathrm{e}=-20$dBm, $\kappa=2.1$, $R_\mathrm{s}=0.1$bits, $p_\mathrm{r_1}=0.3$, $p_\mathrm{r_2}=0.7$, $\alpha=1.2$, $P_\mathrm{R_1}=P_\mathrm{R_2}=P_\mathrm{r_1}=P_\mathrm{r_2}=10$dBm, $m_1\in\{2,4\}$, $m_2\in\{3,10\}$,  $M_1=M_2=M$, \textcolor{black}{ and $\tilde{N}=4.$}

Fig. \ref{fig_asc_power} illustrates the behaviour of the ASC in terms of the transmit power $P$ for both NOMA receiver vehicles $v_\mathrm{r_1}$ and $v_\mathrm{t_2}$ under Fisher-Snedecor $\mathcal{F}$ fading channels. Moreover, the asymptotic behaviour of ASC is presented, which is perfectly match with the exact analytical results in the high SNR regime. 
It can be seen that increasing $P$ provides more ASC for the strong vehicle $v_\mathrm{t_1}$ compared with the weak vehicle  $v_\mathrm{t_2}$, where this advantage becomes increasingly pronounced in the high transmit power. In other words, we can see that as $P$ increases,  $\bar{C}_\mathrm{s}^\mathrm{r_1}$ continuously grows but $\bar{C}_\mathrm{s}^\mathrm{r_2}$ raises first and then  converges to a constant value of $P$. This behaviour is mainly because of SIC in the NOMA principle, where the strong vehicle $v_\mathrm{t_1}$ which takes advantage from SIC has a larger SINR value compared with the weak vehicle $v_\mathrm{t_2}$. However, we can observe that the ASC performance of  the legitimate receiver vehicle pair, i.e., $\bar{C}_\mathrm{s}$, improves as $P$ increases since the main channel conditions becomes better. The performance of ASC versus the number of dual RIS elements $M$ for a given value of the transmit power $P=20$dBm is presented in Fig. \ref{fig_asc_m1}. Although the CLT
approximation is exploited to characterize the small-scale
fading channels, we can observe that the analytical results are perfectly match  with the simulation marks even when $M$ is small. We can clearly see that increasing $M$ can remarkably improve the ASC performance for both receiver vehicles because this increment leads to an improvement in spatial diversity. However, it is observed that raising $M$ does not affect the ASC of $v_\mathrm{t_2}$ after a certain value (e.g., $M=50$) and $\bar{C}_\mathrm{s}^\mathrm{r_2}$  converges to a constant value. The reason behind this behaviour is that when $M$ monotonically increases, i.e., $M\rightarrow\infty$, the SINR of weak receiver vehicle $v_\mathrm{t_2}$ achieves a constant value, i.e., $\gamma_{{\mathrm{r}}_2}\rightarrow\frac{p_\mathrm{r_2}}{p_\mathrm{r_1}}$. 
\begin{figure}[!t]
	\centering
	\includegraphics[width=0.75\columnwidth]{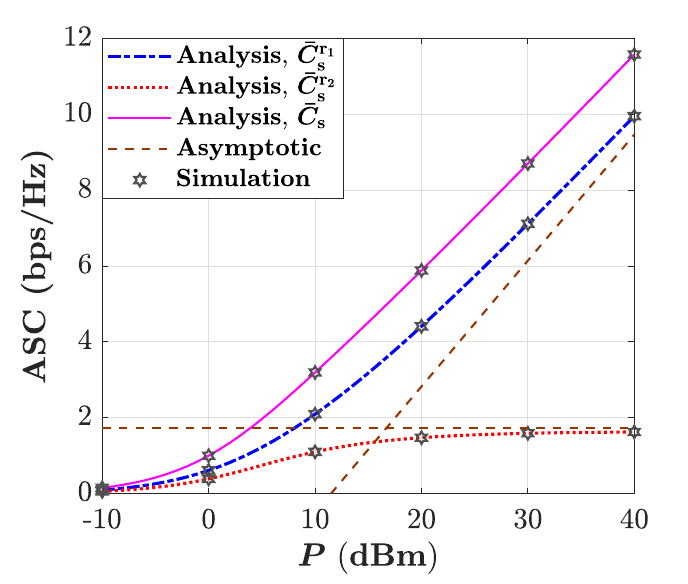}
	\caption{ASC versus the transmit power $P$ when $M=50$.}\vspace{0cm}
	\label{fig_asc_power}
\end{figure}
\begin{figure}[!t]
	\centering
	\includegraphics[width=0.75\columnwidth]{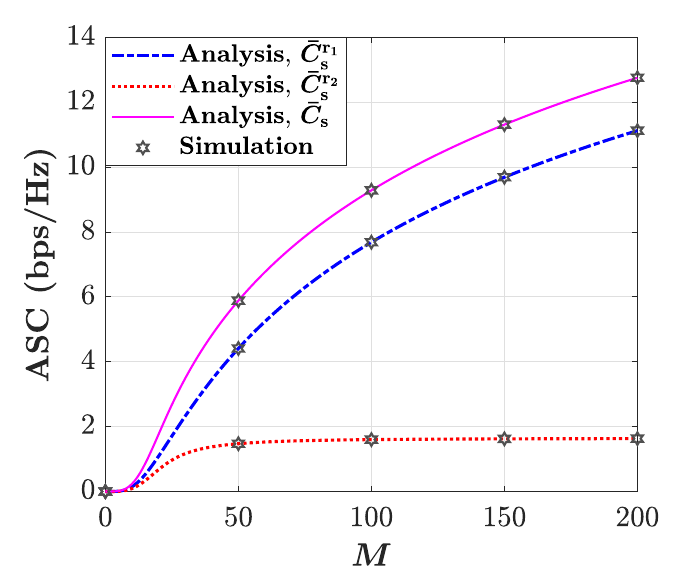}
	\caption{ASC versus the number of dual RIS elements  $M$ when $P=20$dBm.}\vspace{0cm}
	\label{fig_asc_m1}
\end{figure}

Fig. \ref{fig_sop_power} depicts the performance of the SOP in terms of the transmit power $P$ for selected values of $M$ in the considered dual RIS-aided NOMA communication. As expected, we can see that the SOP for both receiver vehicles enhances as $M$ grows since the spatial diversity improves. Further, we can observe that the SOP performance of both strong and weak receiver vehicles initially enhances and then becomes saturated in the medium-to-high transmit power regime. The main reason behind such a behavior for the strong vehicle $v_\mathrm{t_1}$ is that increasing $P$ can simultaneously enhance the received SNR at both $v_\mathrm{t_1}$  and eavesdropper vehicle $v_\mathrm{e}$. Moreover, the reason for the SOP behavior of the weak receiver vehicle $v_\mathrm{t_2}$ is that as $P$ constantly grows, the interference inflicted by $v_\mathrm{t_1}$ over $v_\mathrm{t_2}$ increases, which can dominate the SOP of $v_\mathrm{t_2}$. The impact of fading parameters $m_1$ and $m_2$ on the SOP performance under Fisher-Snedecor $\mathcal{F}$ fading channels is illustrated in Fig. \ref{fig_sop_power_fade}. In the medium-to-high transmit power regime, we can see that as the fading parameters $m_1$ and $m_2$ increase, i.e., the shadowing and fading are less severe, the SOP performance of both weak and strong receiver vehicles improves, which is reasonable since light (heavy) fading decreases (increases) the chance of the received signal dropping below the threshold. Furthermore, we can observe that increasing the fading parameters alone cannot improve the SOP performance in the low transmit power regime. This is because in the low SNR regime, the presence of noise is significant relative to the weak signal, and the channel capacity is basically limited by the weak signal; thereby, even if the shadowed fading severity decreases, the signal may still be too weak to dominate the noise, especially in deep fade conditions.
\begin{figure}[!t]
	\centering
	\includegraphics[width=0.75\columnwidth]{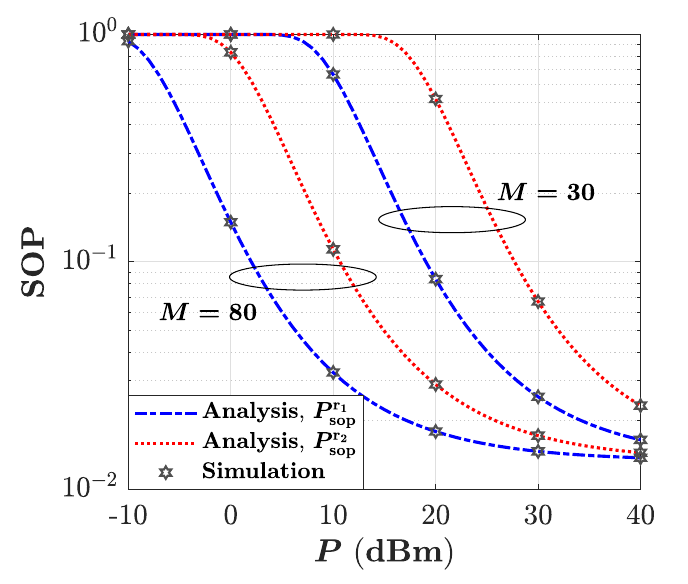}
	\caption{SOP versus the transmit power $P$ for selected values of $M$.}\vspace{0cm}
	\label{fig_sop_power}
\end{figure}

Fig. \ref{fig_ee_m}-(a) reveals the performance of SEE in terms of the number of dual RIS elements $M$ for a given value of the transmit power $P=20$dBm. We can observe that increasing $M$ initially provides more SEE and then steadily reduces the performance of SEE. 
This trend occurs because as $M$ increases, the capability of RISs to control and manipulate the propagation of radio waves improves, which can increase the signal quality; thereby, the SEE performance ameliorates. However, by continuing the increment of $M$, the benefits become limited, and the SEE reduces since the required energy to control the additional elements can overcome the achieved advantages. The behavior of the SEE versus the transmit power $P$ for a given value of $M$ is presented in Fig. \ref{fig_ee_m}-(b). It is clearly seen that the SEE performance initially enhances to its maximum point and then weakens as the transmit power $P$ grows. This behavior is mainly because at first the ASC can significantly dominate the power consumption before reaching the extreme point, but as $P$ becomes larger, the ASC cannot overcome the power consumption anymore; consequently, the SEE drops after the extreme point. 
\begin{figure}[!t]
	\centering
	\includegraphics[width=0.75\columnwidth]{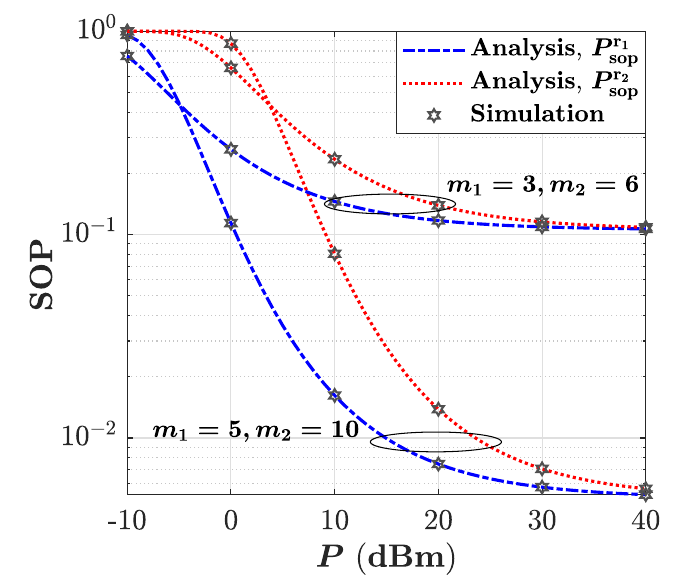}
	\caption{SOP versus the transmit power $P$ for selected values of fading parameters $m_1$ and $m_2$ when $M=80$.  }\vspace{0cm}
	\label{fig_sop_power_fade}
\end{figure}
\begin{figure}[!t]
	\centering
	\includegraphics[width=0.9\columnwidth]{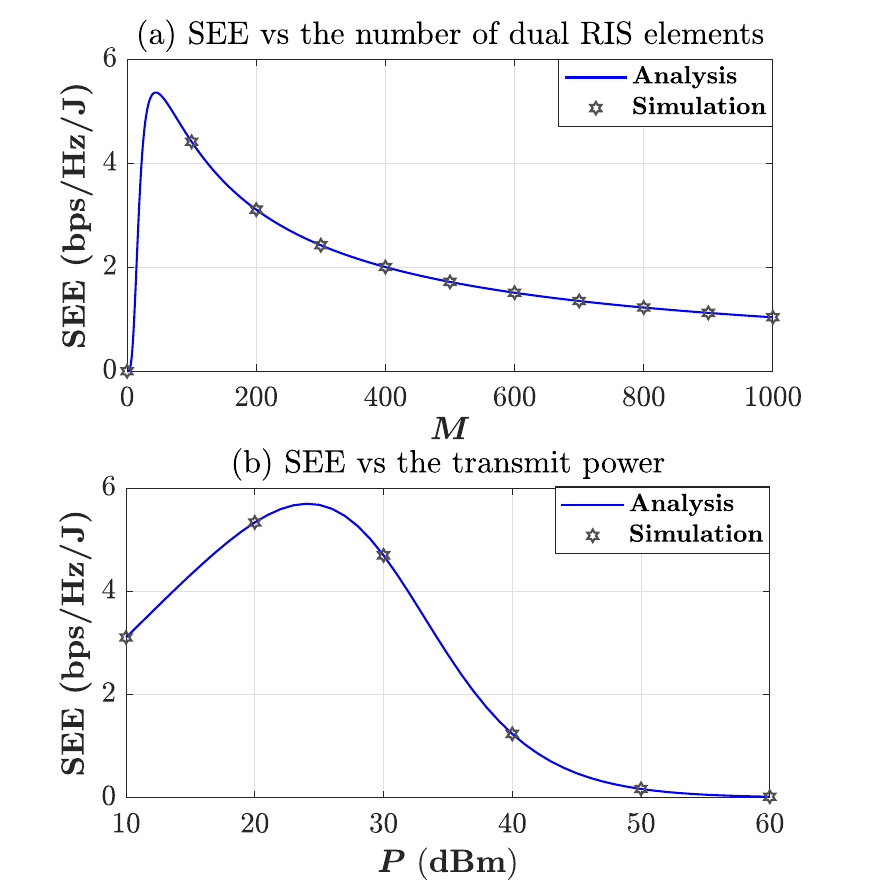}
	\caption{SEE versus (a) the number of dual RIS elements $M$ when $P=20$dBm and (b) the average transmit power $P$ when $M=50$. }\vspace{0cm}
	\label{fig_ee_m}
\end{figure}
\section{Conclusion}\label{sec-con}
This paper examined the performance of PLS over dual RIS-aided V2V NOMA communication systems. Specifically, we assumed that a transmitter vehicle aims to send confidential information to the legitimate receiver vehicles with the help of two RISs, which are positioned in proximity to the transmitter and receivers, respectively, while an eavesdropper vehicle tries to decode the desired signal. 
By assuming that the fading channel between the RISs follows the Fisher-Snedecor $\mathcal{F}$ distribution, we first provided the CDF of SNRs by using the CLT. Then, to assess the secrecy performance, we obtained the closed-form expressions of the ASC, SOP, and SEE by exploiting the Gauss-Laguerre quadrature and the Gaussian quadrature methods. Further, to obtain deeper understanding into the system behavior, we derived the asymptotic expression of the ASC in the high SNR conditions. Eventually, we validated the accuracy of the analytical derivations by Monte Carlo simulation, where the numerical results revealed that applying dual RIS to vehicular networks can remarkably enhance the reliability and security of the transmission in ITS.  
\appendices
\section{Proof of Proposition \ref{pro-asc1}}\label{app-asc1}
Regarding the definition of ASC in \eqref{eq-asc}, the ASC for the strong receiver vehicle $v_\mathrm{r_1}$ can be defined as
\begin{align}
\bar{C}_\mathrm{s}^{\mathrm{r}_1}&=\int_0^\infty\int_0^{\gamma_{{\mathrm{r}}_1}} \log_2\left(\frac{1+\gamma_{{\mathrm{r}}_1}}{1+\gamma_{{\mathrm{e}}}}\right)f_{\gamma_{{\mathrm{r}}_1}}\left(\gamma_{{\mathrm{r}}_1}\right)f_{\gamma_{{\mathrm{e}}}}\left(\gamma_{{\mathrm{e}}}\right)d\gamma_{{\mathrm{e}}}d\gamma_{{\mathrm{r}_1}}\\\notag
&=\int_0^\infty\int_0^\infty\log_2\left(1+\gamma_{{\mathrm{r}}_1}\right)f_{\gamma_{{\mathrm{r}}_1}}\left(\gamma_{{\mathrm{r}}_1}\right)f_{\gamma_{{\mathrm{e}}}}\left(\gamma_{{\mathrm{e}}}\right)d\gamma_{{\mathrm{e}}} d\gamma_{{\mathrm{r}}_1}\\ \notag
&\quad-\int_0^\infty \log_2\left(1+\gamma_{{\mathrm{r}}_1}\right)f_{\gamma_{{\mathrm{r}}_1}}\left(\gamma_{{\mathrm{r}}_1}\right)\bar{F}_{\gamma_{{\mathrm{e}}}}\left(\gamma_{{\mathrm{r}_1}}\right)d\gamma_{{\mathrm{r}_1}} \\
&\quad-\int_0^\infty \log_2\left(1+\gamma_{{\mathrm{r}}_1}\right)f_{\gamma_{{\mathrm{e}}}}\left(\gamma_{{\mathrm{r}}_1}\right)\bar{F}_{\gamma_{{\mathrm{r}_1}}}\left(\gamma_{{\mathrm{r}_1}}\right)d\gamma_{{\mathrm{r}_1}}\\\notag
&=\frac{1}{\ln 2}\int_0^\infty \frac{\bar{F}_{\gamma_{{\mathrm{r}}_1}}\left(\gamma_{{\mathrm{r}}_1}\right)}{1+\gamma_{{\mathrm{r}}_1}}d\gamma_{{\mathrm{r}}_1}\\
&\quad-\frac{1}{\ln 2}\int_0^\infty \frac{\bar{F}_{\gamma_{{\mathrm{r}}_1}}\left(\gamma_{{\mathrm{r}}_1}\right)\bar{F}_{\gamma_{{\mathrm{e}}}}\left(\gamma_{{\mathrm{r}}_1}\right)}{1+\gamma_{{\mathrm{r}}_1}}d\gamma_{{\mathrm{r}}_1}, \label{eq-app1}
\end{align}
where $\bar{F}_{\gamma_{\mathrm{n}}}\left(\gamma_{\mathrm{n}}\right)$ is the complementary CDF (CCDF) of $\gamma_{\mathrm{n}}$. Besides, the first and second terms in \eqref{eq-app1} denote the ASC without eavesdropping and the ASC loss due to the eavesdropper vehicle $v_\mathrm{e}$, respectively. Thus, the ASC of $v_\mathrm{r_1}$ can be re-expressed as
 \begin{align}
	&\bar{C}_\mathrm{s}^{\mathrm{r}_1}=\frac{1}{\ln 2}\int_0^\infty \frac{F_{\gamma_\mathrm{e}}\left(\gamma_{{\mathrm{r}}_1}\right)\bar{F}_{\gamma_{\mathrm{r}_1}}\left(\gamma_{{\mathrm{r}}_1}\right)}{1+\gamma_{{\mathrm{r}}_1}}d\gamma_{{\mathrm{r}}_1}\\
	&=\frac{2}{\ln 2}\int_0^\infty \frac{\left(1-\mathrm{e}^{-\lambda_\mathrm{e}\gamma_{{\mathrm{r}}_1}}\right)}{1+\gamma_{{\mathrm{r}}_1}}\mathrm{erf}\left(\frac{\sqrt{\gamma_{\mathrm{r}_1}}-\mu_U\sqrt{\bar{\gamma}A_1}}{\sigma_U^2\sqrt{2\bar{\gamma}A_1}}\right)d\gamma_{{\mathrm{r}}_1}. \label{eq-app2}
\end{align}
However, it is not straightforward to compute the above integral in a closed-form expression. Instead, we exploit the Gauss-Laguerre quadrature method, which is defined as the following lemma.
\begin{lemma}\label{lemma1}
The Gauss-Laguerre quadrature is define as \cite{abromowitz1972handbook}
\begin{align}
\int_0^\infty \mathrm{e}^{-x}\Lambda(x)dx\approx\sum_{\tilde{n}=1}^{\tilde{N}}w_{\tilde{n}}\Lambda\left(\epsilon_{\tilde{n}}\right)
\end{align}
in which $\epsilon_{\tilde{n}}$ is the $\tilde{n}$-th root of Laguerre polynomial $L_{\tilde{N}}\left(\epsilon_{\tilde{n}}\right)$ and $w_{\tilde{n}}=\frac{\epsilon_{\tilde{n}}}{2\left(\tilde{N}+1\right)^2L^2_{\tilde{N}+1}\left(\epsilon_{\tilde{n}}\right)}$.
\end{lemma}
Therefore, by applying the Gauss-Laguerre quadrature to \eqref{eq-app2} and doing some simplifications, \eqref{eq-asc-v1} is derived and the proof is completed. 
\section{Proof of Proposition \ref{pro-asc2}} \label{app-asc2}
According to the considered system model, the ASC for the weak receiver vehicle $v_\mathrm{r_2}$ is zero when $p_\mathrm{r_2}-\gamma_{{\mathrm{r}}_2}p_\mathrm{r_1}\leq 0$. Hence, following the same approach in proving the ASC of $v_\mathrm{r_1}$, the ASC of $v_\mathrm{r_2}$ can be re-written as  
\begin{align}
	\bar{C}_\mathrm{s}^{\mathrm{r}_2}
	&=\frac{1}{\ln 2}\int_0^\frac{p_\mathrm{r_2}}{p_\mathrm{r_1}} \frac{\bar{F}_{\gamma_{{\mathrm{r}}_2}}\left(\gamma_{{\mathrm{r}}_2}\right)}{1+\gamma_{{\mathrm{r}}_2}}d\gamma_{{\mathrm{r}}_2}\\
	&\quad-\frac{1}{\ln 2}\int_0^\frac{p_\mathrm{r_2}}{p_\mathrm{r_1}} \frac{\bar{F}_{\gamma_{{\mathrm{r}}_2}}\left(\gamma_{{\mathrm{r}}_2}\right)\bar{F}_{\gamma_{{\mathrm{e}}}}\left(\gamma_{{\mathrm{r}}_2}\right)}{1+\gamma_{{\mathrm{r}}_2}}d\gamma_{{\mathrm{r}}_2}\\
	&= \frac{1}{\ln 2}\int_0^\frac{p_\mathrm{r_2}}{p_\mathrm{r_1}} \frac{F_{\gamma_\mathrm{e}}\left(\gamma_{{\mathrm{r}}_2}\right)}{1+\gamma_{{\mathrm{r}}_2}}\left[1-F_{\gamma_{\mathrm{r}_2}}\left(\gamma_{{\mathrm{r}}_2}\right)\right]d\gamma_{{\mathrm{r}}_2}.\label{eq-app3}
\end{align}
Similarly, it is difficult to solve the integral in \eqref{eq-app3}, and thus, we utilize the following lemma to obtain the ASC of $v_\mathrm{r_2}$.
\begin{lemma}
	The Gaussian quadrature is define as \cite{abromowitz1972handbook}
	\begin{align}
		\int_a^b \Delta(x)dx\approx\frac{b-a}{2}\sum_{\tilde{n}=1}^{\tilde{N}}w_{\tilde{n}}\Delta\left(\frac{b-a}{2}\tilde{n}+\frac{b+a}{2}\right)
	\end{align}
	in which $\chi_{\tilde{n}}$ is the $\tilde{n}$-th root of  Legendre polynomial $\mathcal{P}_{\tilde{n}}\left(\chi_{\tilde{n}}\right)$ and $q_{\tilde{n}}=\frac{2}{\left(1-\chi_{\tilde{n}}^2\right)\left[\mathcal{P}'_{\tilde{n}}\left(\chi_{\tilde{n}}\right)\right]^2}$.
\end{lemma}
Now, by applying the Gaussian quadrature to \eqref{eq-app3} and doing some simplifications, the proof is accomplished.
\section{Proof of Corollary \ref{cor-asc1}}\label{app-asc-asy1}
Regarding the definition of ASC, when $\bar{\gamma}\rightarrow\infty$,  \eqref{eq-asc} for $v_\mathrm{t_1}$ can be simplified as 
\begin{align}
&\bar{C}_{\mathrm{s}, \infty}^{\mathrm{r}_1}=\mathbb{E}\left[C_\mathrm{r_1}\right]-\mathbb{E}\left[C_\mathrm{e}\right],
\end{align}
in which the first term can be approximated as follows
\begin{align}
	\mathbb{E}\left[C_\mathrm{r_1}\right]\approx\mathbb{E}\left[\log_2\left( \bar{\gamma}A_1U^2\right)\right]=\log_2\left(\bar{\gamma}A_1\mu_{U^2}\right),
	\end{align}
	where $\mu_{U^2}=\sigma_U^2+\mu_U^2$ and it can be obtained as \eqref{eq-mu2} by considering \eqref{eq-var} and \eqref{eq-mean}. The second term can be estimated as
\begin{align}
	\mathbb{E}\left[C_\mathrm{e}\right]\approx\frac{1}{\ln 2}\int_0^\infty \frac{\bar{F}_{\gamma_{{\mathrm{e}}}}\left(\gamma_{{\mathrm{r}}_1}\right)}{1+\gamma_{{\mathrm{r}}_1}}d\gamma_{{\mathrm{r}}_1},\label{app-3-1}
\end{align} 
where the integral in \eqref{app-3-1} can be computed by applying the integration by parts method, and thus, the proof is completed. 

Following the same method, \eqref{eq-asc} at the high SNR regime for $v_\mathrm{r_2}$ can be given by $\bar{C}_{\mathrm{s}, \infty}^{\mathrm{r}_2}=\mathbb{E}\left[C_\mathrm{r_2}\right]-\mathbb{E}\left[C_\mathrm{e}\right]$, in which the first term can be approximated as $\mathbb{E}\left[C_\mathrm{r_2}\right]\approx\log_2\left(1+\frac{p_\mathrm{r_2}}{p_\mathrm{r_1}}\right)$ and the second term is similar to \eqref{app-3-1} by replacing $\gamma_{{\mathrm{r}}_2}$ instead of $\gamma_{{\mathrm{r}}_1}$. Thus, by solving the obtained integral for $\mathbb{E}\left[C_\mathrm{e}\right]$ with the help of the integration by parts method, the proof is accomplished. \vspace{-1cm}

\section{Proof of Proposition \ref{pro-sop1}} \label{app-sop1}
By substituting the instantaneous secrecy capacity of strong receiver vehicle $v_\mathrm{r_1}$ into \eqref{eq-sop}, the SOP of $v_\mathrm{r_1}$ can be expressed as
\begin{align}
	&P_\mathrm{sop}^{\mathrm{r}_1}=\Pr\left(\log_2\left(\frac{1+\gamma_{\mathrm{r}_1}}{1+\gamma_{\mathrm{e}}}\right)\leq R_\mathrm{s}\right)\\
	&=\Pr\left(\gamma_{\mathrm{r}_1}\leq R_\mathrm{th}\gamma_{\mathrm{e}}+\bar{R}_\mathrm{th}\right)\\
	&=\int_0^\infty F_{\gamma_\mathrm{r_1}}\left(R_\mathrm{th}\gamma_{\mathrm{e}}+\bar{R}_\mathrm{th}\right)f_{\gamma_{\mathrm{e}}}\left(\gamma_\mathrm{e}\right)d\gamma_{\mathrm{e}} \label{eq-app4-1}\\ \notag 
	&= \frac{\lambda_{\mathrm{e}}}{2}\int_0^\infty \mathrm{e}^{-\lambda_\mathrm{e}\gamma_{\mathrm{e}}}\\
	&\quad \times \left[1+\mathrm{erf}\left(\frac{\sqrt{R_\mathrm{th}\gamma_{\mathrm{e}}+\bar{R}_\mathrm{th}}-\sqrt{\bar{\gamma}A_1}\mu_U}{\sigma_U^2\sqrt{2\bar{\gamma}A_1}}\right)\right]d\gamma_{\mathrm{e}}, \label{eq-app4}
\end{align}
Next, by applying the Gauss-Laguerre quadrature to \eqref{eq-app4}, the SOP of $v_\mathrm{r_1}$ can be derived as \eqref{eq-sop-v1} and the proof will be accomplished. 
\section{Proof of Proposition \ref{pro-sop2}}\label{app-sop2}
Regarding the definition of SOP, we can express the SOP of weak receiver vehicle $v_\mathrm{r_2}$ as follows
\begin{align}
	P_\mathrm{sop}^{\mathrm{r}_2}&=\Pr\left(\log_2\left(\frac{1+\gamma_{\mathrm{r}_2}}{1+\gamma_{\mathrm{e}}}\right)\leq R_\mathrm{s}\right)\\
	&=\Pr\left(\gamma_{\mathrm{r}_2}\leq R_\mathrm{th}\gamma_{\mathrm{e}}+\bar{R}_\mathrm{th}\right)\\
	&=\int_0^\infty F_{\gamma_\mathrm{r_2}}\left(R_\mathrm{th}\gamma_{\mathrm{e}}+\bar{R}_\mathrm{th}\right)f_{\gamma_{\mathrm{e}}}\left(\gamma_\mathrm{e}\right)d\gamma_{\mathrm{e}}\\
	&=\frac{\lambda_{\mathrm{e}}}{2}\int_0^\infty \mathrm{e}^{-\lambda_\mathrm{e}\gamma_{\mathrm{e}}}\\
	&\times \left[1+\mathrm{erf}\left(\frac{\sqrt{\frac{R_\mathrm{th}\gamma_{\mathrm{e}}+\bar{R}_\mathrm{th}}{\bar{\gamma}A_2-\bar{\gamma}A_1\left(R_\mathrm{th}\gamma_{\mathrm{e}}+\bar{R}_\mathrm{th}\right)}}-\mu_U}{\sigma_U^2\sqrt{2}}\right)\right]d\gamma_{\mathrm{e}}. \label{eq-app7}
\end{align}
Similarly, by considering the Gauss-Laguerre quadrature, the SOP of $v_\mathrm{v_2}$ can be obtained as \eqref{eq-sop-v2} and the proof is completed. 
\bibliographystyle{IEEEtran}
\bibliography{sample.bib}

\begin{thebibliography}{10}
\providecommand{\url}[1]{#1}
\csname url@samestyle\endcsname
\providecommand{\newblock}{\relax}
\providecommand{\bibinfo}[2]{#2}
\providecommand{\BIBentrySTDinterwordspacing}{\spaceskip=0pt\relax}
\providecommand{\BIBentryALTinterwordstretchfactor}{4}
\providecommand{\BIBentryALTinterwordspacing}{\spaceskip=\fontdimen2\font plus
\BIBentryALTinterwordstretchfactor\fontdimen3\font minus
  \fontdimen4\font\relax}
\providecommand{\BIBforeignlanguage}[2]{{%
\expandafter\ifx\csname l@#1\endcsname\relax
\typeout{** WARNING: IEEEtran.bst: No hyphenation pattern has been}%
\typeout{** loaded for the language `#1'. Using the pattern for}%
\typeout{** the default language instead.}%
\else
\language=\csname l@#1\endcsname
\fi
#2}}
\providecommand{\BIBdecl}{\relax}
\BIBdecl

\bibitem{guo2022vehicular}
H.~Guo, X.~Zhou, J.~Liu, and Y.~Zhang, ``{Vehicular intelligence in 6G:
  Networking, communications, and computing},'' \emph{Veh. Commun.}, vol.~33,
  p. 100399, 2022.

\bibitem{wang2021green}
J.~Wang, K.~Zhu, and E.~Hossain, ``{Green Internet of Vehicles (IoV) in the 6G
  era: Toward sustainable vehicular communications and networking},''
  \emph{IEEE trans. green commun. netw.}, vol.~6, no.~1, pp. 391--423, 2021.

\bibitem{noor20226g}
M.~Noor-A-Rahim, Z.~Liu, H.~Lee, M.~O. Khyam, J.~He, D.~Pesch, K.~Moessner,
  W.~Saad, and H.~V. Poor, ``{6G for vehicle-to-everything (V2X)
  communications: Enabling technologies, challenges, and opportunities},''
  \emph{Proc. IEEE}, vol. 110, no.~6, pp. 712--734, 2022.

\bibitem{kaveh2020lightweight}
M.~Kaveh, D.~Mart{\'\i}n, and M.~R. Mosavi, ``{A lightweight authentication
  scheme for V2G communications: A PUF-based approach ensuring cyber/physical
  security and identity/location privacy},'' \emph{Electronics}, vol.~9, no.~9,
  p. 1479, 2020.

\bibitem{huang2022non}
Z.~Huang, L.~Bai, X.~Cheng, X.~Yin, P.~E. Mogensen, and X.~Cai, ``{A
  non-stationary 6G V2V channel model with continuously arbitrary
  trajectory},'' \emph{IEEE Trans. Veh. Technol.}, vol.~72, no.~1, pp. 4--19,
  2022.

\bibitem{basar2019wireless}
E.~Basar, M.~Di~Renzo, J.~De~Rosny, M.~Debbah, M.-S. Alouini, and R.~Zhang,
  ``{Wireless communications through reconfigurable intelligent surfaces},''
  \emph{IEEE Access}, vol.~7, pp. 116\,753--116\,773, 2019.

\bibitem{basharat2021reconfigurable}
S.~Basharat, S.~A. Hassan, H.~Pervaiz, A.~Mahmood, Z.~Ding, and M.~Gidlund,
  ``{Reconfigurable intelligent surfaces: Potentials, applications, and
  challenges for 6G wireless networks},'' \emph{IEEE Wirel. Commun.}, vol.~28,
  no.~6, pp. 184--191, 2021.

\bibitem{mu2021simultaneously}
X.~Mu, Y.~Liu, L.~Guo, J.~Lin, and R.~Schober, ``{Simultaneously transmitting
  and reflecting (STAR) RIS aided wireless communications},'' \emph{IEEE Trans.
  Wirel. Commun.}, vol.~21, no.~5, pp. 3083--3098, 2021.

\bibitem{palomares2023enabling}
A.~Palomares-Caballero, C.~Molero, F.~R. Ghadi, F.~J. L{\'o}pez-Mart{\'\i}nez,
  P.~Padilla, D.~Morales-Jimenez, and J.~F. Valenzuela-Vald{\'e}s, ``{Enabling
  intelligent omni-surfaces in the polarization domain: Principles,
  implementation and applications},'' \emph{IEEE Commun. Mag.}, vol.~61,
  no.~11, pp. 144--150, 2023.

\bibitem{liu2021star}
Y.~Liu, X.~Mu, J.~Xu, R.~Schober, Y.~Hao, H.~V. Poor, and L.~Hanzo, ``{STAR:
  Simultaneous transmission and reflection for 360° coverage by intelligent
  surfaces},'' \emph{IEEE Wirel. Commun.}, vol.~28, no.~6, pp. 102--109, 2021.

\bibitem{lu2018survey}
Z.~Lu, G.~Qu, and Z.~Liu, ``{A survey on recent advances in vehicular network
  security, trust, and privacy},'' \emph{IEEE Trans. Intell. Transp. Syst.},
  vol.~20, no.~2, pp. 760--776, 2018.

\bibitem{shiu2011physical}
Y.-S. Shiu, S.~Y. Chang, H.-C. Wu, S.~C.-H. Huang, and H.-H. Chen, ``{Physical
  layer security in wireless networks: A tutorial},'' \emph{IEEE Wirel.
  Commun.}, vol.~18, no.~2, pp. 66--74, 2011.

\bibitem{zhang2021physical}
J.~Zhang, H.~Du, Q.~Sun, B.~Ai, and D.~W.~K. Ng, ``{Physical layer security
  enhancement with reconfigurable intelligent surface-aided networks},''
  \emph{IEEE Trans. Inf. Forensics Secur.}, vol.~16, pp. 3480--3495, 2021.

\bibitem{vega2022physical}
J.~D. Vega-S{\'a}nchez, G.~Kaddoum, and F.~J. L{\'o}pez-Mart{\'\i}nez,
  ``{Physical layer security of ris-assisted communications under
  electromagnetic interference},'' \emph{IEEE Commun. Lett.}, vol.~26, no.~12,
  pp. 2870--2874, 2022.

\bibitem{ghadi2022ris}
F.~R. Ghadi and W.-P. Zhu, ``{RIS-aided secure communications over
  Fisher-Snedecor F fading channels},'' \emph{arXiv preprint arXiv:2208.07274},
  2022.

\bibitem{kaveh2023secrecy}
M.~Kaveh, Z.~Yan, and R.~J{\"a}ntti, ``{Secrecy Performance Analysis of
  RIS-Aided Smart Grid Communications},'' \emph{IEEE Trans. Ind. Inform.},
  2023.

\bibitem{tang2021novel}
Z.~Tang, T.~Hou, Y.~Liu, J.~Zhang, and C.~Zhong, ``{A novel design of RIS for
  enhancing the physical layer security for RIS-aided NOMA networks},''
  \emph{IEEE Wirel. Commun. Lett.}, vol.~10, no.~11, pp. 2398--2401, 2021.

\bibitem{saito2013non}
Y.~Saito, Y.~Kishiyama, A.~Benjebbour, T.~Nakamura, A.~Li, and K.~Higuchi,
  ``{Non-orthogonal multiple access (NOMA) for cellular future radio access},''
  in \emph{2013 IEEE 77th Veh. Technol. Conf. (VTC Spring)}.\hskip 1em plus
  0.5em minus 0.4em\relax IEEE, 2013, pp. 1--5.

\bibitem{di2017v2x}
B.~Di, L.~Song, Y.~Li, and Z.~Han, ``{V2X meets NOMA: Non-orthogonal multiple
  access for 5G-enabled vehicular networks},'' \emph{IEEE Wirel. Commun.},
  vol.~24, no.~6, pp. 14--21, 2017.

\bibitem{gu2022socially}
X.~Gu, W.~Duan, G.~Zhang, Y.~Ji, M.~Wen, and P.-H. Ho, ``{Socially aware V2X
  networks with RIS: Joint resource optimization},'' \emph{IEEE Trans. Veh.
  Technol.}, vol.~71, no.~6, pp. 6732--6737, 2022.

\bibitem{chen2021robust}
Y.~Chen, Y.~Wang, and L.~Jiao, ``{Robust transmission for reconfigurable
  intelligent surface aided millimeter wave vehicular communications with
  statistical CSI},'' \emph{IEEE Trans. Wirel. Commun.}, vol.~21, no.~2, pp.
  928--944, 2021.

\bibitem{chapala2023intelligent}
V.~K. Chapala and S.~Zafaruddin, ``Intelligent connectivity through
  ris-assisted wireless communication: Exact performance analysis with phase
  errors and mobility,'' \emph{IEEE trans. intell. veh.}, 2023.

\bibitem{kavaiya2023restricting}
S.~Kavaiya and D.~K. Patel, ``{Restricting passive attacks in 6G vehicular
  networks: a physical layer security perspective},'' \emph{Wirel. Netw.},
  vol.~29, no.~3, pp. 1355--1365, 2023.

\bibitem{mensi2022performance}
N.~Mensi and D.~B. Rawat, ``{On the performance of partial RIS selection vs.
  partial relay selection for vehicular communications},'' \emph{IEEE Trans.
  Veh. Technol.}, vol.~71, no.~9, pp. 9475--9489, 2022.

\bibitem{ai2021secure}
Y.~Ai, A.~Felipe, L.~Kong, M.~Cheffena, S.~Chatzinotas, and B.~Ottersten,
  ``{Secure vehicular communications through reconfigurable intelligent
  surfaces},'' \emph{IEEE Trans. Veh. Technol.}, vol.~70, no.~7, pp.
  7272--7276, 2021.

\bibitem{gu2022physical}
X.~Gu, W.~Duan, G.~Zhang, Q.~Sun, M.~Wen, and P.-H. Ho, ``{Physical layer
  security for RIS-aided wireless communications with uncertain eavesdropper
  distributions},'' \emph{IEEE Systems Journal}, vol.~17, no.~1, pp. 848--859,
  2022.

\bibitem{zhang2023star}
Y.~Zhang, Z.~Yang, J.~Cui, P.~Xu, G.~Chen, Y.~Wu, and M.~Di~Renzo, ``{STAR-RIS
  Assisted Secure Transmission for Downlink Multi-Carrier NOMA Networks},''
  \emph{IEEE Trans. Inf. Forensics Secur.}, 2023.

\bibitem{do2021multi}
T.~N. Do, G.~Kaddoum, T.~L. Nguyen, D.~B. Da~Costa, and Z.~J. Haas,
  ``{Multi-RIS-aided wireless systems: Statistical characterization and
  performance analysis},'' \emph{IEEE Trans. Commun.}, vol.~69, no.~12, pp.
  8641--8658, 2021.

\bibitem{zhao2021cooperative}
Y.~Zhao, W.~Xu, H.~Sun, D.~W.~K. Ng, and X.~You, ``{Cooperative reflection
  design with timing offsets in distributed multi-RIS communications},''
  \emph{IEEE Wirel. Commun. Lett.}, vol.~10, no.~11, pp. 2379--2383, 2021.

\bibitem{phan2022performance}
V.-D. Phan, B.~C. Nguyen, T.~M. Hoang, T.~N. Nguyen, P.~T. Tran, B.~V. Minh,
  and M.~Voznak, ``{Performance of cooperative communication system with
  multiple reconfigurable intelligent surfaces over Nakagami-m fading
  channels},'' \emph{IEEE Access}, vol.~10, pp. 9806--9816, 2022.

\bibitem{kumar2023performance}
P.~Kumar, A.~Bhowmick, and Y.~K. Choukiker, ``{Performance of
  Multi-RIS-Assisted D2D Communication using NOMA},'' \emph{IEEE Access}, 2023.

\bibitem{nguyen2023performance}
B.~C. Nguyen, N.~T. Xuan, H.~T. Manh, H.~L.~T. Thanh, and P.~T. Hiep,
  ``{Performance analysis for multi-RIS UAV NOMA mmWave communication
  systems},'' \emph{Wirel. Netw.}, vol.~29, no.~2, pp. 761--773, 2023.

\bibitem{tran2022exploiting}
P.~T. Tran, B.~C. Nguyen, T.~M. Hoang, X.~H. Le \emph{et~al.}, ``{Exploiting
  multiple RISs and direct link for performance enhancement of wireless systems
  with hardware impairments},'' \emph{IEEE Trans. Commun.}, vol.~70, no.~8, pp.
  5599--5611, 2022.

\bibitem{ma2022cooperative}
X.~Ma, Y.~Fang, H.~Zhang, S.~Guo, and D.~Yuan, ``{Cooperative beamforming
  design for multiple RIS-assisted communication systems},'' \emph{IEEE Trans.
  Wireless Commun.}, vol.~21, no.~12, pp. 10\,949--10\,963, 2022.

\bibitem{zhang2023double}
P.~Zhang, S.~Gong, and S.~Ma, ``{Double-RIS aided multi-user MIMO
  communications: Common reflection pattern and joint beamforming design},''
  \emph{IEEE Trans. Veh. Technol.}, 2023.

\bibitem{xue2023multi}
Q.~Xue, R.~Wei, S.~Ma, Y.~Xu, and L.~Yan, ``{Multi-user mmWave uplink
  communications based on collaborative double-RIS: Joint beamforming and power
  control},'' \emph{IEEE Commun. Lett.}, 2023.

\bibitem{ma2023power}
H.~Ma, H.~Wang, H.~Zhao, and S.~Fu, ``{Power Minimization for Double
  Cooperative-RIS-Assisted Uplink NOMA System},'' \emph{IEEE Wirel. Commun.
  Lett.}, 2023.

\bibitem{shaikh2022performance}
M.~H.~N. Shaikh, K.~Rabie, X.~Li, T.~Tsiftsis, and G.~Nauryzbayev, ``{On the
  Performance of Dual RIS-assisted V2I Communication under Nakagami-m
  Fading},'' in \emph{2022 IEEE 96th Veh. Technol. Conf. (VTC2022-Fall)}.\hskip
  1em plus 0.5em minus 0.4em\relax IEEE, 2022, pp. 1--5.

\bibitem{ghadi2023performance}
F.~R. Ghadi, M.~Kaveh, and D.~Mart{\'\i}n, ``{Performance Analysis of
  RIS/STAR-IOS-aided V2V NOMA/OMA Communications over Composite Fading
  Channels},'' \emph{IEEE trans. intell. veh.}, 2023.

\bibitem{huang2019reconfigurable}
C.~Huang, A.~Zappone, G.~C. Alexandropoulos, M.~Debbah, and C.~Yuen,
  ``{Reconfigurable intelligent surfaces for energy efficiency in wireless
  communication},'' \emph{IEEE Trans. Wirel. Commun.}, vol.~18, no.~8, pp.
  4157--4170, 2019.

\bibitem{bjornson2019intelligent}
E.~Bj{\"o}rnson, {\"O}.~{\"O}zdogan, and E.~G. Larsson, ``{Intelligent
  reflecting surface versus decode-and-forward: How large surfaces are needed
  to beat relaying?}'' \emph{IEEE Wireless Commun. Lett.}, vol.~9, no.~2, pp.
  244--248, 2019.

\bibitem{yoo2017fisher}
S.~K. Yoo, S.~L. Cotton, P.~C. Sofotasios, M.~Matthaiou, M.~Valkama, and G.~K.
  Karagiannidis, ``{The Fisher--Snedecor $\mathcal{F}$ distribution: A simple
  and accurate composite fading model},'' \emph{IEEE Commun. Lett.}, vol.~21,
  no.~7, pp. 1661--1664, 2017.

\bibitem{sanchez2020physical}
J.~D.~V. S{\'a}nchez, P.~Ram{\'\i}rez-Espinosa, and F.~J.
  L{\'o}pez-Mart{\'\i}nez, ``{Physical layer security of large reflecting
  surface aided communications with phase errors},'' \emph{IEEE Wireless
  Commun. Lett.}, vol.~10, no.~2, pp. 325--329, 2020.

\bibitem{abromowitz1972handbook}
M.~Abromowitz and I.~A. Stegun, ``Handbook of mathematical functions,'' 1972.

\end{thebibliography}

\end{document}